\documentclass[11pt]{article}
\usepackage[margin=1in]{geometry}
\usepackage{complexity}
\usepackage{mathtools,amsmath,amssymb}
\usepackage{xspace}
\usepackage{comment}
\RequirePackage[colorlinks=true]{hyperref}
\hypersetup{
  linkcolor=[rgb]{0,0,0.5},
  citecolor=[rgb]{0, 0.5, 0},
  urlcolor=[rgb]{0.5, 0, 0}
}
\usepackage{dsfont}

\usepackage{amsthm}
\usepackage{thmtools,thm-restate}

\numberwithin{equation}{section}
\declaretheoremstyle[bodyfont=\it,qed=\qedsymbol]{noproofstyle}

\declaretheorem[numberlike=equation]{observation}

\declaretheorem[name=Observation,numbered=no]{observation*}

\declaretheorem[numberlike=equation]{fact}

\declaretheorem[numberlike=equation]{theorem}

\declaretheorem[name=Theorem,numbered=no]{theorem*}

\declaretheorem[numberlike=equation]{lemma}
\declaretheorem[name=Lemma,numbered=no]{lemma*}

\declaretheorem[numberlike=equation]{corollary}
\declaretheorem[name=Corollary,numbered=no]{corollary*}

\declaretheorem[name=Proposition,numbered=no]{proposition*}

\declaretheorem[numberlike=equation]{claim}
\declaretheorem[name=Claim,numbered=no]{claim*}

\declaretheorem[numberlike=equation]{conjecture}
\declaretheorem[name=Conjecture,numbered=no]{conjecture*}

\declaretheorem[numberlike=equation]{question}
\declaretheorem[name=Question,numbered=no]{question*}

\declaretheoremstyle[bodyfont=\it,qed=$\lozenge$]{defstyle} 

\declaretheorem[numberlike=equation,style=defstyle]{definition}
\declaretheorem[unnumbered,name=Definition,style=defstyle]{definition*}

\declaretheorem[unnumbered,name=Example,style=defstyle]{example*}

\declaretheorem[unnumbered,name=Notation=defstyle]{notation*}

\declaretheorem[unnumbered,name=Construction,style=defstyle]{construction*}

\declaretheorem[numberlike=equation,style=defstyle]{remark}
\declaretheorem[unnumbered,name=Remark,style=defstyle]{remark*}

\usepackage{nth}
\usepackage{intcalc}
\usepackage{etoolbox}
\usepackage{xstring}
\hypersetup{
}

\usepackage{ifpdf}
\ifpdf
\else
\usepackage[quadpoints=false]{hypdvips}
\fi

\newcommand{\ECCCurl}[2]{http://eccc.weizmann.ac.il/report/\ifnumcomp{#1}{>}{93}{19}{20}#1/#2/}

\newcommand{\shortECCC}[2]{\texttt{\href{http://eccc.weizmann.ac.il/report/\ifnumcomp{#1}{>}{93}{19}{20}#1/#2/}{eccc:TR#1-#2}}}

\newcommand{\parseECCC}[1]{
\StrSubstitute{#1}{TR}{}[\tmpstring]%
\IfSubStr{\tmpstring}{/}{ 
\StrBefore{\tmpstring}{/}[\ecccyear]%
\StrBehind{\tmpstring}{/}[\ecccreport]%
}{
\StrBefore{\tmpstring}{-}[\ecccyear]%
\StrBehind{\tmpstring}{-}[\ecccreport]%
}%
\shortECCC{\ecccyear}{\ecccreport}}

	\renewcommand{\vec}[1]{{\mathbf{#1}}}

	\makeatletter
	\newcommand{\va}{{\vec{a}}\@ifnextchar{^}{\!\:}{}}
	\newcommand{\vb}{{\vec{b}}\@ifnextchar{^}{\!\:}{}}
	\newcommand{\vc}{{\vec{c}}\@ifnextchar{^}{\!\:}{}}
	\newcommand{\vd}{{\vec{d}}\@ifnextchar{^}{\!\:}{}}
	\newcommand{\ve}{{\vec{e}}\@ifnextchar{^}{\!\:}{}}
	\newcommand{\vy}{{\vec{y}}\@ifnextchar{^}{\!\:}{}}
	\newcommand{\vs}{{\vec{s}}\@ifnextchar{^}{\!\:}{}}
	\newcommand{\vt}{{\vec{t}}\@ifnextchar{^}{\!\:}{}}
	\newcommand{\vx}{{\vec{x}}\@ifnextchar{^}{}{}}		
	\newcommand{\vz}{{\vec{z}}\@ifnextchar{^}{\!\:}{}}
	\newcommand{\vu}{{\vec{u}}\@ifnextchar{^}{\!\:}{}}	
	
	\newcommand{\vY}{{\vec{Y}}\@ifnextchar{^}{\!\:}{}}
	\newcommand{\vX}{{\vec{X}}\@ifnextchar{^}{}{}}		
	\newcommand{\vZ}{{\vec{Z}}\@ifnextchar{^}{\!\:}{}}
	\newcommand{\vG}{{\vec{G}}\@ifnextchar{^}{\!\:}{}}
	
	\makeatother

\newcommand{\F}{\mathbb{C}}
\newcommand{\N}{\mathbb{N}}

 \newcommand{\margcomm}[2]{\marginpar{\tiny{\textbf{#1 }\textit{#2}}}}
\newcommand{\nutan}[1]{\margcomm{}{}}
\newcommand{\srikanth}[1]{\margcomm{}{}}
\newcommand{\chandra}[1]{\margcomm{}{}}
\newcommand{\mrinal}[1]{\margcomm{}{}}
\newcommand{\prasad}[1]{\margcomm{}{}}
\newcommand{\adrian}[1]{\margcomm{}{}} 

\newcommand{\MKnote}[1]{\textcolor{blue}{{\sf Mrinal :} {#1}}}

\newcommand{\set}[1]{\left\{#1\right\}}

\newcommand{\abs}[1]{\left|#1\right|}

\DeclareMathOperator{\rank}{rank}

\def\epsilon{\varepsilon}

\newcommand*\samethanks[1][\value{footnote}]{\footnotemark[#1]}
\date{}

\title{Schur Polynomials do not have small formulas if the Determinant doesn't! }
\author{
Prasad Chaugule\thanks{Dept. of Computer Science \& Engineering, IIT Bombay. Email: \texttt{prasad|mrinal|nutan|ckm@cse.iitb.ac.in}. A part of the work of the second author was done during a postdoctoral stay at the University of Toronto. The third author is  supported by MATRICS grant MTR/2017/000909 awarded by SERB, Government
of India.}
\and 
Mrinal Kumar \samethanks
\and Nutan Limaye\samethanks
\and Chandra Kanta Mohapatra \samethanks
\and Adrian She\thanks{Dept. of Computer Science, University of Toronto. Email : \texttt{ashe@cs.toronto.edu}.}
\and Srikanth Srinivasan\thanks{Dept. of Mathematics, IIT Bombay. Email :  \texttt{srikanth@math.iitb.ac.in}. Supported by MATRICS grant MTR/2017/000958 awarded by SERB, Government
of India.}
}
\begin{document}
\maketitle

\begin{abstract}
Schur Polynomials are families of symmetric polynomials that have been classically studied in Combinatorics and Algebra alike. They play a central role in the study of Symmetric functions, in  Representation theory~\cite{StanleyVol2}, in Schubert calculus~\cite{ledoux2010introductory} as well as in Enumerative combinatorics~\cite{gasharov1996incomparability, stanley1984number, StanleyVol2}. In recent years, they have also shown up in various incarnations in Computer Science, e.g, Quantum computation~\cite{hallgren2000normal, odonnell2015quantum} and Geometric complexity theory~\cite{ikenmeyer2017rectangular}. 

However, unlike some other families of symmetric polynomials like the Elementary Symmetric polynomials, the Power Symmetric polynomials and the Complete Homogeneous Symmetric polynomials, the computational complexity of syntactically computing Schur polynomials has not been studied much. In particular, it is not known whether  Schur polynomials can be computed efficiently by algebraic formulas. In this work, we address this  question, and show that unless \emph{every} polynomial with a small algebraic branching program (ABP) has a small algebraic formula, there are Schur polynomials that cannot be computed by algebraic formula of polynomial size. In other words, unless the algebraic complexity class $\mathrm{VBP}$ is equal to the complexity class $\mathrm{VF}$, there exist Schur polynomials which do not have polynomial size  algebraic formulas. 

As a consequence of our proof, we also show  that computing the determinant of certain \emph{generalized} Vandermonde matrices is essentially as hard as computing the general symbolic determinant. To the best of our knowledge, these are one of the first hardness results of this kind for families of polynomials which are not \emph{multilinear}. A key ingredient of our proof is the study of  composition of \emph{well behaved} algebraically independent polynomials with a homogeneous polynomial, and  might be of independent interest. 
\end{abstract}

\thispagestyle{empty}
\newpage
\pagenumbering{arabic}

\section{Introduction}
In this paper, we explore a theme at the intersection of \emph{Algebraic Complexity Theory}, which studies the computational complexity of computing multivariate polynomials using algebraic operations, and \emph{Algebraic Combinatorics}, which studies, among other things, algebraic identities among polynomials associated to various combinatorial objects. 

Specifically, the questions we study are related to the computational complexity of \emph{Symmetric Polynomials}, which are polynomials in $\mathbb{C}[x_1,\ldots,x_n]$ that are invariant under permutations of the underlying variable set $x_1,\ldots,x_n$.\footnote{In Combinatorics literature, these are more commonly known as \emph{Symmetric Functions}. One can also consider symmetric functions over fields other than the complex numbers, but throughout this paper, we will stick to $\mathbb{C}.$} Examples of such polynomials include
\begin{itemize}
\item The \emph{Elementary Symmetric polynomials} $e_0,e_1,\ldots, e_n$ where $e_d = \sum_{|S| = d}\prod_{j\in S} x_j$ is the sum of all multilinear monomials of degree exactly $d$,
\item The \emph{Complete Homogeneous Symmetric polynomials} $h_0,h_1,\ldots$ where $h_d$ is the sum of all monomials (multilinear or otherwise) of degree exactly $d$, and
\item The \emph{Power Symmetric polynomials} $p_0,p_1,\ldots$ where $p_d = \sum_{i = 1}^n x_i^d.$
\end{itemize}
It is a standard fact that the above three families generate \emph{all} symmetric polynomials in a well-defined sense. More precisely, the Fundamental Theorem of Symmetric Polynomials states that every  symmetric polynomial $f$ can be written \emph{uniquely} as a polynomial in $\{e_1,\ldots,e_n\},$ and similarly in $\{h_1,\ldots,h_n\}$ and $\{p_1,\ldots,p_n\},$ each of which is thus an algebraically independent set of polynomials. In particular, for $\lambda = (\lambda_1,\lambda_2,\ldots,\lambda_\ell)$ a non-increasing sequence of positive integers, if we define $e_\lambda = \prod_{i\in [\ell]}e_{\lambda_i},$ then the $\{e_\lambda\}_{\lambda}$ are linearly independent, and moreover the set $E_d:=\{e_\lambda\ |\ \sum_i \lambda_i = d\}$ forms a basis for the vector space $\Lambda_d$ of homogeneous symmetric polynomials of degree $d$; the same is true also of $h_\lambda$ and $p_\lambda$ (defined analogously), yielding bases $H_d$ and $P_d$ respectively for $\Lambda_d$.

\paragraph*{Symmetric Polynomials in Mathematics.} The study of Symmetric Polynomials is a classical topic in Mathematics, with close connections to combinatorics and representation theory, among other fields (see, e.g., \cite{MACD,sagan}). In representation theory , it is known that the entries of the change-of-basis matrices between different bases for the space $\Lambda_d$ yield important numerical invariants of various representations of the symmetric group $S_d$. In algebraic and enumerative combinatorics, the study of symmetric polynomials leads to formulas and generating functions for interesting combinatorial quantities such as \emph{Plane Partitions} (see, e.g.~\cite[Chapter 7]{StanleyVol2}). These studies have in turn given rise to a dizzying array of algebraic identities and generating functions for various families of symmetric polynomials. Some of these, as we note below, have already had consequences for computational complexity.

\paragraph*{Algebraic Complexity of Symmetric Polynomials.} Symmetric polynomials have also been intensively investigated by researchers in Algebraic complexity~\cite{nw1997,sw2001,Shpilka,HY,FLMS17}, with several interesting consequences. The famous `Ben-Or trick' in algebraic complexity (also known simply as `interpolation') was discovered by Ben-Or~\cite{sw2001} in the context of using a standard generating function for the Elementary Symmetric Polynomials to obtain small depth-$3$ formulas for $e_1,\ldots,e_n.$\footnote{The generating function referred to is $\prod_{i=1}^n(t - x_i) = \sum_{j=0}^n (-1)^{n-j}e_{n-j} t^j.$} The same idea also yields small constant-depth formulas for the complete homogeneous symmetric polynomials. Symmetric polynomials have also been used to prove lower bounds for several interesting models of computation including homogeneous and inhomogeneous $\Sigma\Pi\Sigma$ formulas~\cite{nw1997,sw2001, Shpilka}, homogeneous multilinear formulas~\cite{HY} and homogenous $\Sigma\Pi\Sigma\Pi$ formulas~\cite{FLMS17}. Further, via reductions, the elementary and power symmetric polynomials have been used to define restricted models of algebraic computation known as the symmetric circuit model~\cite{Shpilka} and the $\Sigma\wedge\Sigma$ model (or Waring rank), which in turn have been significantly investigated (see, e.g.~\cite{sax08, L10, oeding16}). 
\paragraph*{Schur polynomials.} In this paper, we study the complexity of an important family of symmetric polynomials called the \emph{Schur Polynomials}, which we now define.

\begin{definition}
\label{def:Schur-Jacobi}
Let $\lambda = (\lambda_1,\ldots,\lambda_\ell)$ be a non-increasing sequence of positive integers with $\sum_i \lambda_i = d.$ We define the Schur polynomial $s_\lambda(x_1,\ldots,x_n)$ of degree $d$ as follows. 
\[s_\lambda = \frac{\mathrm{det}\left((x_i^{\lambda_j+n-j})_{i,j\in [n]}\right)}{\mathrm{det}\left((x_i^{n-j})_{i,j\in [n]}\right)}\]
(Here, if $\lambda = (\lambda_1,\ldots,\lambda_\ell),$ then we define $\lambda_j = 0$ for $j > \ell$.)
\end{definition}

The Schur polynomials are known to generalize the elementary symmetric polynomials as well as homogeneous symmetric polynomials. It is also known that the Schur polynomials of degree $d$ form a basis for $\Lambda_d$, which the vector space of all homogeneous symmetric polynomials of degree $d$. 

The Schur polynomials occupy a central place in the study of symmetric polynomials.  Their importance in representation theory can be seen for instance by the fact that, they describe the characters of representations of the general linear and symmetric groups. In particular, consider the general linear group $\text{GL}(V)$ over a complex vector space $V$ of dimension $n$. If $\rho$ is an irreducible representation of $\text{GL}(V)$ that is polynomial, meaning that the eigenvalues of $\rho(A)$ can be expressed as a polynomial in the eigenvalues of $A$, then the character $\text{Tr}(\rho(A))$ is a Schur polynomial $s_\lambda(x_1, \dots, x_n)$ evaluated at the eigenvalues $x_1, \dots, x_n$ of $A$, where $\lambda$ is a partition with at most $n$ non-zero parts. 
Furthermore, the entries of the change-of-basis matrix (for the vector space $\Lambda_d$) from the power symmetric polynomials to the Schur polynomials are exactly the values of the irreducible characters of the symmetric group $S_d$. Specifically, \emph{the Murnaghan-Nakayama rule} states that when expanded into the basis of power symmetric polynomials we have $$s_\lambda = \sum_{\mu = (\mu_1, \dots, \mu_l) \vdash n} \chi^\lambda(\mu) \prod_{i=1}^l p_{\mu_i}$$ where $\chi^\lambda(\mu)$ is an irreducible character of the symmetric group $S_n$ evaluated at a permutation of cycle type $\mu$. See \cite[Chapter 7]{StanleyVol2} for further details about these uses in representation theory.

Beyond representation theory, Schur polynomials are used in algebraic geometry in the Schubert calculus \cite{ledoux2010introductory}, which is used to calculate the number of ways in which Schubert subvarieties in the Grassmannian (the set of all $k$ dimensional linear subspaces in an $n$-dimensional space) may intersect. Schur polynomials are also used in enumerative combinatorics, as they provide generating functions for counting various combinatorial objects, including plane partitions, tableaux \cite[Chapter 7]{StanleyVol2}, reduced decompositions of permutations \cite{stanley1984number}, and graph colourings \cite{gasharov1996incomparability}. 

Being one of the most well-studied objects in the theory of symmetric functions, Schur polynomials appear in many different avatars in the literature. The following classical definition is also known to capture Schur polynomials. The definition uses combinatorial structures called Ferrers diagrams. A Ferrers diagram (or a Young diagram or simply a diagram) of shape $\lambda$, is a left-aligned two-dimensional array of boxes with the $i$th row containing $\lambda_i$ many boxes. (See, e.g. Stanley~\cite{StanleyVol2}, for more about Ferrers diagrams.)

\begin{definition}
\label{def:schurpolys}
Consider a Ferrers diagram of shape $\lambda$. For any non-decreasing sequence $\mu = (\mu_1,\ldots,\mu_m)$ with $\sum_j \mu_j = d$, we define the \emph{Kostka number} $K_{\lambda\mu}$ to be the number of ways of filling the boxes of the Ferrers diagram with numbers from $1,\ldots,m$ such that each row is non-decreasing, each column is strictly increasing, and the number of $i$'s equals $\mu_i$ for each $i\in [m]$.

The Schur polynomial $s_{\lambda}(x_1,\ldots,x_n)\in \Lambda_d$ is defined so that the coefficient of $x_1^{\mu_1}\cdots x_m^{\mu_m}$ is the Kostka number $K_{\lambda\mu}$ (the coefficients of other monomials are defined by symmetry). 

In particular, $s_\lambda = 0$ if $n < \ell.$ So we assume that $n\geq \ell$ throughout.
\end{definition}

From this definition it is easy to see that Schur polynomials generalize both elementary symmetric polynomials (when $\ell = d$ and $\lambda_1 = \lambda_2\cdots = \lambda_d = 1$ in the definition above) and homogeneous symmetric polynomials  (when $\ell = 1$ and $\lambda_1 = d$ in the definition above).

The Kostka numbers used in the definition above have been investigated extensively both from combinatorial and computational perspectives. (See for instance~\cite{StanleyVol2,narayanan2006complexity}.)



\paragraph*{Algebraic Complexity of Schur Polynomials.} In this work we focus on the algebraic complexity of Schur polynomials. As stated in \autoref{def:Schur-Jacobi}, which is also known as the \emph{bialternant formula} of Jacobi, the Schur polynomial $s_\lambda$ can be expressed as the ratio of two determinants. In particular, this implies that the Schur polynomials have  algebraic circuits of size $\poly(n,d)$. In fact, it also implies that these polynomials belong to the smaller algebraic complexity class $\mathrm{VBP}$,\footnote{These is class of polynomial families which can be efficiently computed by algebraic branching programs.} for which the Determinant is the complete polynomial. 

However, this upper bound is quite a bit weaker than what is known for other well-studied symmetric polynomials mentioned above, all of which have \emph{constant-depth} formulas of polynomial size. We consider the question of whether the Schur polynomials have constant-depth formulas of polynomial size or even general (arbitrary depth) formulas of polynomial size. 

Our main result is that under reasonable complexity assumptions, the answer to the above question is negative for many different $\lambda.$ (Note that since the elementary and complete homogeneous symmetric polynomials are particular examples of Schur polynomials, there are \emph{some} Schur polynomials that have formulas of polynomial size.)

\begin{theorem}[Main Theorem]
\label{thm:main-intro}
Assume that $\lambda= (\lambda_1,\ldots,\lambda_\ell)$ is such that $\lambda_i \geq \lambda_{i+1}+(\ell-1)$ for all $i \in [\ell-1]$, also $\lambda_{\ell} \geq \ell$ and let $d = \sum_i \lambda_i$. Then, for $n \geq \lambda_1 + \ell$, if $s_\lambda(x_1,\ldots,x_n)$ has an algebraic formula of size $s$ and depth $\Delta$, then the $\ell\times \ell$ determinant $(\mathrm{det}_\ell )$ has an algebraic formula of size $\poly(s)$ and depth $\Delta + O(1).$
\end{theorem}

For suitable choices of $\ell,d,n$ above, we can ensure that these parameters are all polynomially related. The theorem then implies that the Schur polynomials do not have algebraic formulas of polynomial size unless the entire complexity class $\mathrm{VBP}$ collapses to the complexity class $\mathrm{VF}$ which consists of polynomials with small formulas. Moreover, the Schur polynomials do not have  \emph{constant-depth} formulas of subexponential size unless the determinant does, a result that would greatly improve the state-of-the-art in this direction~\cite{GKKS16}. 

The above theorem and its proof have several interesting aspects that we now elaborate on. 

\paragraph*{Newton iteration and formula complexity.} \autoref{thm:main-intro} is motivated in part by a recent result of Bl\"{a}ser and Jindal~\cite{blser_et_al:LIPIcs:2018:10140} who prove the following interesting result about symmetric polynomials. As mentioned earlier, it is known that any symmetric polynomial $f_{sym}\in \mathbb{C}[x_1,\ldots,x_n]$ can be written \emph{uniquely} as a polynomial in (say) the elementary symmetric polynomials $e_1,\ldots,e_n$. I.e., there exists a unique $f_{E}\in \mathbb{C}[x_1,\ldots,x_n]$ such that  \chandra{We are not so consistent with the notation of the symmetric polynomial. We should fix it to either $f$ or $f_{sym}$.  }
\[
f_{sym}(x_1,\ldots,x_n) = f_E(e_1,\ldots,e_n).
\] 
Motivated by a question of Regan and Lipton~ \cite{Lipton2009Reganblog}, Bl\"{a}ser and Jindal studied the computational complexity of $f_{sym}$ vis-a-vis that of $f_E.$ It is clear that if $f_E$ has  algebraic circuits of polynomial size (resp. formulas) then so does $f_{sym}$, since the elementary symmetric polynomials have algebraic formulas of polynomial size. Interestingly, Bl\"{a}ser and Jindal showed a converse to this statement: they showed that if $f_{sym}$ has small algebraic circuits, then so does $f_E.$~\footnote{Bl\"{a}ser and Jindal work throughout with the elementary symmetric polynomials. However, using algebraic identities that link various symmetric polynomials with each other, we observe in this paper that their result also holds for the complete homogeneous and power symmetric polynomials.}

At first sight, this looks highly relevant to our theorem, since by the classical \emph{Jacobi-Trudi} identity (see, e.g.~\autoref{thm:jacobi-trudi} in this paper or \cite[Theorem 7.16.1]{StanleyVol2}), when $f_{sym}$ is a Schur polynomial of the type assumed in \autoref{thm:main-intro}, then $f_{E}$ is in fact the determinant (on a subset of its variables). We could hope to use the theorem of Bl\"{a}ser and Jindal to prove that if the Schur polynomial has a small formula, then so does the determinant. However, this doesn't quite work, since the proof of~\cite{blser_et_al:LIPIcs:2018:10140} only yields small \emph{circuits} for the polynomial $f_E$, even if we assume that the polynomial $f_{sym}$ has small \emph{formulas.}

We briefly outline the reason for this, noting that this hurdle occurs quite often in trying to adapt results in algebraic circuit complexity to algebraic formulas. As mentioned above, the polynomials $e_1,\ldots,e_n$ are algebraically independent. A standard proof of this (see, e.g.,~\cite{Shpilka}) goes via showing that the map
\[
\overline{e}:\mathbb{C}^n\rightarrow \mathbb{C}^n,\text{ defined by }\overline{a} = (a_1,\ldots,a_n) \mapsto (e_1(\overline{a}),\ldots,e_n(\overline{a}))
\]
is surjective. Hence, for each $\overline{b}\in \mathbb{C}^n,$ there exists an $\overline{a}\in \mathbb{C}^n$ such that $\overline{e}(\overline{a}) = \overline{b}.$ The reason this is relevant to the result of~\cite{blser_et_al:LIPIcs:2018:10140} is that if we have an efficient algorithm for `inverting' $\overline{e}$ in this way and we additionally have an efficient algorithm for computing $f_{sym}$, then we immediately obtain an efficient algorithm for computing $f_E$ on any given input $\overline{b}\in \mathbb{C}^n$ by first inverting the map $\overline{e}$ to obtain an $\overline{a}$ as above, and then applying the algorithm for $f_{sym}$ to obtain $f(\overline{a}) = f_E(\overline{b}).$ The main technical result in Bl\"{a}ser and Jindal's work is to show how to invert the map $\overline{e}$ as above using an algebraic circuit. The inversion is done by carefully applying a standard algebraic version of \emph{Newton iteration}, which can be performed by an efficient algebraic circuit. Having done this, we plug the output of this circuit into the circuit for $f_{sym}$ to obtain the circuit for $f_E$.

The reason the above proof does not work in the setting of algebraic formulas is the use of Newton iteration, which is not known to be doable  with small formulas (or even within the seemingly larger class $\mathrm{VBP}$). Indeed, this is the main bottleneck in translating several results in algebraic complexity on polynomial factorization~\cite{k89, DSS17, cks_general} and hardness-randomness tradeoffs~\cite{ki03, DSY09, cks19} that are known in the context of algebraic circuits to the setting of algebraic formulas. 

In the proof of the main theorem, we show how to get around the use of Newton iteration in this setting and use it to prove a (slightly weaker) version of the result of Bl\"{a}ser and Jindal for algebraic formulas. We hope that the ideas we use here can be adapted and extended to circumvent the use of Newton iteration in some of the other settings mentioned above as well. Our main technical lemma is the following.

\begin{lemma}[Main Technical Lemma (informal)]\label{lem:key-intro}
Let $g_1,\ldots,g_n\in \mathbb{C}[x_1,\ldots,x_n]$ be ``well-behaved'' algebraically independent polynomials. Then, for any homogeneous polynomial $\tilde{f}$, if $f = \tilde{f}(g_1,\ldots,g_n)$ has a formula of size $s$ and depth $\Delta$, the polynomial $\tilde{f}$ has a formula of size $\poly(s)$ and depth $\Delta + O(1).$
\end{lemma}
For a formal definition of what we mean by ``well-behaved'' and for a formal statement of this lemma, we refer the reader to \autoref{def:prop s} and \autoref{lem:key} respectively. This lemma, and some of the ideas in its (very simple) proof might be of independent interest and may have other applications, e.g. in \autoref{sec:partial derivatives} we discuss an application of this lemma to some special cases of a question of Amir Shpilka on proving lower bounds on the partial derivative complexity of a product of algebraically independent polynomials.

\paragraph*{Generalized Vandermonde determinants.} The Vandermonde matrix $(x_i^{n-j})_{i,j\in [n]}$ and its determinant are ubiquitous in computation because of their relation to polynomial interpolation. More precisely, the problem of finding a degree-$(n-1)$ univariate polynomial that takes prescribed values at a specified set of $n$ distinct points involves solving a linear system of equations where the underlying matrix is precisely the Vandermonde matrix. It is, therefore, an important fact that the Vandermonde determinant is computationally much easier than the general determinant: in fact, it has the following standard depth-$2$ formula
\[
\mathrm{det}_n\left((x_i^{n-j})_{i,j\in [n]}\right) = \prod_{i,j\in [n]: i < j}(x_i - x_j).
\]
However, it is unclear whether such small expressions continue to exist if we allow the exponents of the variables to vary more generally. For integers $\mu_1 > \mu_2 > \cdots > \mu_n \geq 0$, consider the \emph{generalized Vandermonde matrix} $(x_i^{\mu_j})_{i,j\in [n]}$. Similar to the Vandermonde matrix, the determinant of this matrix is related to the problem of \emph{sparse polynomial interpolation}, where we are trying to interpolate a polynomial only involving  the monomials of degree $\mu_1,\ldots,\mu_n$ through the given points. 

Can we expect that computing any generalized Vandermonde determinant is much easier than computing the determinant itself? It follows from \autoref{thm:main-intro} and the bialternant formula from \autoref{def:Schur-Jacobi} above that the answer to this question is negative: for certain (polynomially large) exponents, the generalized Vandermonde determinant is not much easier than the determinant.

\srikanth{Changed $\vx$ to $x_1,\ldots,x_n$ for next paragraph to homogenize notation for introduction.}
\srikanth{Move this paragraph to earlier in the intro?}

\paragraph*{Discussion on Schur Polynomials and Generating functions}

In algebraic and enumerative combinatorics, we often study a family of related combinatorial  objects by considering a \emph{generating function} that combines them, in the hope that the generating function yields a nice closed-form expression which can further be used to estimate or otherwise understand these objects better. (See e.g.~\cite{Wilf,Flajolet} for much more about this.) For instance, we know that the generating functions for the elementary and complete homogeneous symmetric polynomials
\[
E(t) = \sum_{i=0}^n t^i e_i(\vx) \phantom{xxxxxx}\text{and}\phantom{xxxxxx} H(t) = \sum_{i= 0}^n t^i h_i(\vx)
\]
have small expressions given by
\[
\phantom{xx} E(t) = \prod_{i\in [n]}(1+tx_i) \phantom{xxxxxxx}\text{and}\phantom{xxxxxx} H(t) = \frac{1}{\prod_{i\in [n]} (1-tx_i)}.
\]
Furthermore, as such expressions are algebraic formulas using additions, multiplications and divisions, we can use these formulas along with division elimination and interpolation to construct small algebraic formulas for the $e_i$s and $h_j$s themselves.

Recall that both the elementary and complete homogeneous symmetric polynomials are special cases of Schur polynomials. It therefore is natural to ask if generating functions can be obtained for other simple sequences of Schur polynomials. Our results imply that the generating function for certain sequences of Schur polynomials do \emph{not} have small closed-form expressions with small formulas unless the determinant has small formulas. This seems like an interesting statement in algebraic combinatorics, conditioned upon a well-known conjecture in Computational Complexity theory.

For concreteness, here is one such `hard' generating function made up of Schur polynomials. For any $\ell \geq 0,$ let $\lambda_\ell = (\ell^2,\ell^2-\ell,\ell^2-2\ell,\cdots,\ell).$ Define \chandra{ In later part, we are showing that this $\lambda_{\ell}$ is easy but then how this is hard here? What am I missing? I think we need to add $n-\ell$ extra zeros where $n \geq \ell^2+\ell$ , otherwise $s_{\lambda_\ell}$ has a small formula, which leads to an easy generating function.}
\[
S(t) = \sum_{\ell \geq 0} t^\ell s_{\lambda_\ell}.
\]
Note that this is a finite sum for any fixed $n$ as $s_{\lambda_\ell} = 0$ if $\ell > n.$ In algebraic combinatorics, it is common to consider symmetric polynomials in an \emph{infinite} number of variables in which case the above is truly an infinite sum. A simple expression in the infinite case typically leads to a simple expression in the finite case by simply setting all variables other than $x_1,\ldots,x_n$ to $0$ in the expression.

\paragraph*{Proving hardness of non-multilinear polynomial families. }
The most natural and widely studied notion of completeness in algebraic setting is the notion of  projections. A polynomial $P \in \F[x_1,\ldots,x_n]$ is said to be a projection of a polynomial $Q\in \F[y_1,\ldots,y_m]$ if there is setting $\sigma$ of $y_1,\ldots,y_m$ to either field constants or to variables from the set $\{x_1,\ldots,x_n\}$, such that the polynomial $Q\left(\sigma(y_1), \sigma(y_2), \ldots, \sigma(y_m)\right)$ equals $P$. While this notion of reductions is very natural and intuitive and in particular, it is clear that \emph{easiness} of $Q$ (with respect to having a small algebraic circuit or formula, for instance) immediately implies the \emph{easiness} of $P$, there is an inherent difficulty in using this notion of reductions for proving the hardness of families of non-multilinear polynomials. To see this, observe that if $Q$ is non-multilinear in each of its variables, and $P$ is a multilinear polynomial which depends on at least one variable, then $P$ cannot be expressed as a projection of $Q$. In particular, this notion of reductions cannot be used to prove the hardness of non-multilinear Schur polynomials or the hardness of  generalized Vandemonde determinant,  assuming the hardness of determinant for algebraic formulas. We avoid this issue by showing that given a small formula for one of these candidate non-multilinear hard polynomials, we can come up with a \emph{small} formula for the Determinant, with only a mild increase in size. This step does more than projections, and there is a slight increase in the formula size in the process. This argument is more in the spirit of Turing reductions in standard Computational Complexity.

\paragraph*{Other related work.} 


The algebraic complexity of Schur polynomials has been studied in various restricted models of computation. Koev~\cite{K07}, Chan et al.~\cite{CDEKK19} and Fomin et al.~\cite{FGK16} consider the complexity of computing Schur polynomials in \emph{the subtraction-free} algebraic circuit model. An algebraic circuit is \emph{subtraction-free} if it uses only addition, multiplication and division operators.\footnote{For example, consider the polynomial $x^2-xy+y^2$. It is computed by the following subtraction-free circuit: $(x^3+y^3)/(x+y)$.} 
They showed that $s_\lambda(x_1,\ldots, x_n)$ has subtraction-free circuits of size polynomial in $n$ and $\lambda_1$. In Fomin et al.~\cite{FGK16}, the authors also proved polynomial bounds on the size of the subtraction-free circuits computing other interesting variants of Schur polynomials such as \emph{double Schur polynomials} and \emph{skew Schur polynomials}. All the algorithms presented in~\cite{K07,CDEKK19,FGK16} for computing Schur polynomials used division in non-trivial ways. 

Demmel et al.~\cite{DK06} and Fomin et al.~\cite{FGNS18} studied the monotone complexity of Schur polynomials. In the monotone setting, only addition and multiplication operators are used. (Both division and subtraction operators are not allowed.) 
They proved exponential upper bounds on the monotone complexity of Schur polynomials and conjectured an exponential lower bound. The exact complexity of Schur polynomials is not resolved in the monotone setting. 
However, Grigoriev et al.~\cite{GK16} proved an  exponential monotone circuit lower bound for a related family of symmetric polynomials, called the \emph{monomial symmetric polynomials}. 

\paragraph*{Organization of the paper. } The rest of the paper is organized as follows. We start with a brief discussion of some of the preliminaries in \autoref{sec:prelims} and a brief introduction to Symmetric polynomials and Schur polynomials in \autoref{sec:prelim-sym-poly}. We formally state and prove \autoref{lem:key-intro} in \autoref{sec:key lemma}, followed by its application to the proof of \autoref{thm:main-intro} in \autoref{sec:schur-formula-complexity}. We discuss further applications of some of these ideas to extending the result of Bl\"{a}ser and Jindal's~\cite{blser_et_al:LIPIcs:2018:10140} in \autoref{sec:BJ} and to the question of proving lower bounds on the partial derivative complexity of a product of algebraically independent polynomials in \autoref{sec:partial derivatives}.  We conclude with some open questions in \autoref{sec:open q}.

\section{Notations and Preliminaries}
\label{sec:prelims}
Throughout this paper, we assume that we are working over the field $\F$. It is not very hard to see that the results we present can be made to work for fields of characteristic zero or fields of sufficiently large characteristic. Boldface letters are used for tuples of variables e.g. $\vx$ for $(x_1, x_2, \ldots, x_n)$. For $\vb = (b_1, b_2, \ldots, b_n) \in \N^n$ and $\vx = (x_1, x_2, \ldots, x_n)$, we use  $\vx^{\vb}$ to denote  $\prod_{i = 1}^n x_i^{b_i}$. 
We use $|\vb|_1$ to denote $\sum_{i \in [n]} b_i$. 

\subsection{Models of computation}
We start by defining some of the standard models of algebraic computation that we work with in the rest of the paper. 
\begin{definition}[Algebraic circuit]	
	An algebraic circuit $C$ is a directed acyclic graph usually having two kinds of node $+,\times$ with a unique sink vertex having out-degree $0$ called the root. The source vertices having in-degree $0$ are labelled by either field constants or formal variables. Further edges entering into $+$ gate can have field constants on them allowing the plus gate to compute the $\mathbb{F}$-linear combination of its children. The root outputs the polynomial computed by the algebraic circuit. 
\end{definition}

\begin{definition}[Algebraic formula]
	If the underlying graph is a tree instead of a directed acyclic graph then the circuit is called a formula. Otherwise we can define formula to be a circuit having output fan-in at most $1$ for every node.
\end{definition}

\begin{definition}[Algebraic branching program]
	Algebraic branching program (ABP) is a layered graph having a unique source vertex (we call it $s$) and a unique sink vertex (we call it $t$). All the edges are from layer $i$ to $i+1$, and each edge is labelled by a linear polynomial. The weight of a path $p$ is the product of the labels over the edges in $p$. The polynomial that the ABP computes is the sum of all weighted paths from $s$ to $t$. 
\end{definition}

\subsection{Interpolation and Division elimination}
We now state two fairly standard facts about algebraic formula.  The first of these relates the formula complexity of a polynomial to the formula complexity of each of its homogeneous components.  
\begin{lemma}\label{lem:homogeneous components}
Let $P(\vx) \in \F[\vx]$ be a polynomial which can be computed by a formula of size at most $s$ and depth $\Delta$. Then, for every $d$, the homogeneous component of $P$ of degree $d$ can be computed by a formula of size at most $O(s^2)$ and depth $\Delta+O(1)$.  
\end{lemma}
The proof of this lemma is via a standard interpolation argument, where we consider the polynomial $Q(t) = P(x_1t, x_2t, \ldots, x_nt) \in \F(\vx)[t]$ as a univariate in $t$. The point to note is that the homogeneous components of $P$ are coefficients of various powers of $t$ in this new polynomial, and hence can be computed as a linear combination of sufficiently many evaluations of $Q(t)$ for distinct values of $t$ in the base field (which we assume to be large enough). For every $\alpha \in \F$, the formula size of $Q(\alpha)$ is upper bounded by the formula size of $P$. Similarly the depth of the formula of $Q(\alpha)$ is bounded by the depth of $P$. The number of such distinct evaluations needed is upper bounded by one more than the degree of $Q$, which is one more than the degree of $P$ itself. The final observation needed for proving the size upper bound is that a polynomial which can be computed by a formula of size $s$ has degree upper bounded by $s$. Thus, we need to take an appropriately weighted linear combination of $s+1$ distinct substitutions of $t$ in $Q$, each of which has a formula of size at most $s$; thereby giving us an upper bound of $O(s^2)$. Taking linear combinations of such substitutions can be done in depth $O(1)$, which gives the overall depth bound of $\Delta+O(1)$. 

The next statement we need is about the formula complexity of a polynomial which can be written as quotient of two polynomials with small formulas. 
\begin{lemma}\label{lem:division of formulas}
Let $P$ and $R$ be polynomials in $\F[\vx]$ of formula (ABP/circuit) size at most $s$ and depth at most $\Delta$ such that $R$ divides $P$. Then, the polynomial $Q = \frac{P}{R}$ can be computed by a formula (an ABP/circuit resp.) of size at most $\poly(s)$ and depth at most $\Delta+O(1)$. 
\end{lemma}
The proof of this lemma goes via the standard division elimination argument of Strassen and that of~\autoref{lem:homogeneous components}. We refer the reader to the excellent survey of Shpilka and Yehudayoff~\cite{sy} for formal details on division elimination.



\subsection{Algebraic independence and the Jacobian}
The notion of algebraic independence that we now define plays a crucial role in the proofs in the paper. We start with a formal definition. 
\begin{definition}\label{def:algebraic independence}
Polynomials $q_1, q_2, \ldots, q_k \in  \F[\vx]$ are said to be algebraically independent over $\F$ if there is no non-zero polynomial $g(y_1, y_2, \ldots, y_k) \in \F[\vy]$ such that $g(q_1, q_2, \ldots, q_k)$ is identically zero. 
\end{definition}
This definition generalizes the notion of \emph{linear} independence, which is the special case when there is no non-zero polynomial $g$ in $k$ variables and degree $1$ such that $g(q_1, q_2, \ldots, q_k)$ is identically zero. As we shall see next, over fields of characteristic zero ( or sufficiently large characteristic), the notion of algebraic independence is characterized by the rank of the Jacobian matrix defined below. 
\begin{definition}\label{def:jacobian}
The Jacobian matrix of a tuple $(q_1, q_2, \ldots, q_k)$ of $n$ variate polynomials in $\F[\vx]$, denoted by ${\cal J}(q_1, q_2, \ldots, q_k)$ is a $k \times n$ matrix with entries from $\F[\vx]$ whose $(i, j)^{th}$ entry equals $\frac{\partial q_i}{\partial x_j}$. 
\end{definition}
Thus, the row corresponding to $q_i$ in ${\cal J}(q_1, q_2, \ldots, q_k)$ contains all of the $n$ first order partial derivatives of $q_i$. In other words, the $i^{th}$ row of the Jacobian gives us the gradient of $q_i$. The connection between algebraic independence and the Jacobian stems from the following (almost folklore) theorem. 
\begin{theorem}[Jacobian and Algebraic Independence]\label{thm:jac and alg ind}
Let $\left(q_1, q_2, \ldots, q_k \right)$ be a $k$ tuple of $n$ variate polynomials in $\F[\vx]$ of degree at most $d$. Then, $q_1, q_2, \ldots, q_k$ are algebraically independent over $\F$ if and only if the the rank of the Jacobian matrix ${\cal J}(q_1, q_2, \ldots, q_k)$ over the field $\F(\vx)$ is equal to $k$.
\end{theorem}
A proof of this theorem can be found in the survey of Chen, Kayal and Wigderson~\cite[Chapter 3]{ckw}.

\subsection{Taylor's expansion}
\label{sec:prelim-taylor}
For our proof, we need the following well-known form of Taylor's expansion. 
\begin{theorem}\label{thm:taylor}
Let $P \in \F[\vx]$ be an $n$ variate polynomial of degree at most $d$, and  let $\va \in \F^n$ be a point. Then, for an $n$-tuple of variables $\vz$
\[
P(\va + \vz) = \sum_{i = 0}^d \left(\sum_{\vu \in \N^{n}, |\vu|_1 = i} \frac{\vz^{\vu}}{\vu !} \cdot \frac{\partial P}{\partial \vx^{\vu}}\left(\va\right) \right) \, 
\]  
where, for $\vu = (u_1, u_2, \ldots, u_n)$, $\vu! = u_1!\cdot u_2!\cdots u_n!.$
\end{theorem} 
Note that for $i = 0$, the summand $\left(\sum_{\vu \in \N^{n}, |\vu|_1 = i} \frac{\vz^{\vu}}{\vu !} \cdot \frac{\partial P}{\partial \vx^{\vu}}\left(\va\right) \right)$ is just equal to $P(\va)$, and for every positive integer $i$ at most $d$, this summand is a homogeneous polynomial of degree equal to $i$ in $\vz$. Of particular utility to us is the following easy corollary of~\autoref{thm:taylor}. 
\begin{corollary}\label{cor:taylor truncation}
Let $P \in \F[\vx]$ be an $n$ variate polynomial of degree $d \geq 1$, and  let $\va \in \F^n$ be a point. Then, for an $n$-tuple of variables $\vz$
\[
P(\va + \vz) = P(\va) + \sum_{j = 1}^n z_j \cdot \frac{\partial P}{\partial x_j}(\va) \mod \langle \vz \rangle^2\, .
\]  

\end{corollary}

\subsection{Two useful lemmas}
We use the following (well known) lemma in our arguments. While the lemma is essentially folklore, we sketch a proof for completeness. 
\begin{lemma}\label{lem:remainder}
Let $f(\vx) \in \F[\vx]$ and $P(\vx, y) \in \F[\vx, y]$ be  polynomials such that $P\left(\vx, f(\vx)\right)$  is identically zero. Then, there exists a polynomial $Q(\vx, y) \in \F[\vx, y]$ such that $P(\vx, y) = (y-f(\vx))\cdot Q(\vx, y)$. 
\end{lemma}
\begin{proof}[Proof]
Let $d$ be the degree of $P$ in $y$ and let $C_0(\vx), C_1(\vx), \ldots, C_d(\vx)$ be polynomials in $\F[\vx]$ such that 
\[
P(\vx, y) = \sum_{i = 0}^d C_i(\vx)\cdot y^i \, .
\]
Therefore, $P(\vx, f(\vx))$ can be written as 
\[
P(\vx, f(\vx)) = \sum_{i = 0}^d C_i(\vx)\cdot f(\vx)^i \, .
\]
Subtracting the two expressions above, we get 
\[
P(\vx, y) - P(\vx, f(\vx)) =  \sum_{i = 0}^d C_i(\vx)\cdot y^i  - \sum_{i = 0}^d C_i(\vx)\cdot f(\vx)^i \, .
\]
Now, on the right hand side, the term for $i = 0$ cancels out and on further simplification, we get 
\[
P(\vx, y) - P(\vx, f(\vx)) =  \sum_{i = 1}^d C_i(\vx)\cdot \left(y^i - f(\vx)^i\right)  \, .
\]
Note that for every natural number $i \geq 1$, $y^i-f(\vx)^i$ is divisible by  $(y-f(\vx))$. Therefore, every summand on the right hand side has $(y-f(\vx))$ as a factor, and thus there is a polynomial $Q(\vx, y)$ such that 
\[
P(\vx, y) - P(\vx, f(\vx)) =  (y-f(\vx))\cdot Q(\vx, y)  \, .
\]
Moreover, since $P(\vx, f(\vx))$ is identically zero, we have that $
P(\vx, y) =  (y-f(\vx))\cdot Q(\vx, y)  \, , 
$ thereby completing the proof of the lemma. 
\end{proof}
We now state the well known Polynomial Identity Lemma.\footnote{This lemma is referred to by many names in the literature, e.g. the Schwartz--Zippel Lemma, or the  DeMillo--Lipton--Schwartz--Zippel Lemma and has been discovered multiple times, starting with {\O}ystein Ore in 1922. 
For a brief history, see \cite{arvind}
where the term ``Polynomial Identity Lemma''
is attributed to L.\ Babai.}
\begin{lemma}[Polynomial Identity Lemma]\label{lem:polynomial identity lemma}  
	Let $\F$ be a field, and let $P\in \F[\vx]$ be a non-zero polynomial of degree (at most) $d$ in $n$ variables. Then, for any finite set $S\subset {\F}$ we have
	$$|\{ \va \in S^n: P(\va) = 0\} | \le d {|S|}^{n-1}.$$
\end{lemma}

In particular, if $|S| \ge d+1$, then there exists some $\va \in S^n$ satisfying $P(\va) \ne 0$. 
This gives us a brute force deterministic algorithm, running in time $(d+1)^n$, to test if an arithmetic circuit computing a polynomial of degree at most $d$ in $n$ variables is identically zero.

\section{Symmetric polynomials}
\label{sec:prelim-sym-poly}
A polynomial is said to be symmetric if it is invariant under a permutation of variables. We now define some of the families of symmetric polynomials that are discussed in this paper and briefly discuss some of their properties. For a more detailed introduction on symmetric polynomials, we refer the reader to the book~\cite{MACD}. We start with the definitions.
\begin{definition}[Elementary symmetric polynomials]\label{def:esym}
The elementary symmetric polynomial of degree $k$ on $n$ variables denoted by $e_k(\vx)$ is defined as follows:
\[
e_k(\vx) = \sum_{S\subseteq [n]} \prod_{i \in S} x_i \, .
\]
\end{definition}
The following fact states a property of the elementary symmetric polynomials which will be useful for our proofs in the later sections. 
\begin{fact}\label{fact:esym and roots}
For all $\alpha_1, \alpha_2, \ldots, \alpha_n \in \F$, if $c_1, c_2, \ldots, c_n$ are field elements such that \[
\prod_{i = 1}^n (z-\alpha_i) = z^{n} - c_1\cdot z^{n-1} + c_2\cdot z^{n-2}- \cdots + (-1)^{n} c_n \, ,
\]
then, for every $j \in [n]$, $c_j = e_j(\alpha_1, \alpha_2, \ldots, \alpha_n)$.
\end{fact}

\begin{definition}[Homogeneous (Complete) symmetric polynomials] The homogeneous symmetric polynomial of degree $k$ on $n$ variables denoted by $h_k(\vx)$ is defined as follows:
	$$h_k(\vx)=\sum_{\vb \in \N^n: |\vb|_1 = k} \vx^\vb.$$ 
\end{definition}

\begin{definition}[Power symmetric polynomials] The power symmetric polynomial of degree $k$ on $n$ variables denoted by $p_k(\vx)$ is defined as follows:
	$p_k= \sum_{i=1}^n  x_i^k$.
\end{definition}

These sets of polynomials are algebraically independent. The following fact states this formally. 
\begin{fact}\label{fact:sym ad}
Let $\vx$ be an $n$ tuple of variables. Then, elementary symmetric polynomials $e_1(\vx), \ldots, e_n(\vx)$   are algebraically independent over $\F$. Similarly,  homogeneous symmetric polynomials $h_1(\vx), \ldots, h_n(\vx)$ and power symmetric polynomials $p_1(\vx), p_2(\vx), \ldots, p_n(\vx)$ are also algebraically independent over $\F$. 
\end{fact}


We now formally state the fundamental theorem of symmetric polynomials, which essentially says that over field of characteristic zero, every symmetric polynomial can be written as a unique polynomial combinations of the elementary symmetric polynomials (similarly, for power symmetric polynomials or homogeneous symmetric polynomials).
\begin{theorem}[The fundamental theorem of symmetric polynomials]
\label{thm:fundamental}
For a symmetric polynomial $f_{sym} \in \F[\vx]$ there exists a unique polynomial $f \in \F[\vx]$ s.t $f_{sym}=f_E(e_1(\vx),e_2(\vx), \ldots e_n(\vx))$ where $e_i(\vx)$ is the elementary symmetric polynomial of degree $i$. 
 
Similarly, there exists a unique polynomial $f_H \in \F[\vx]$ such that $f_{sym} =f_H(h_1(\vx),h_2(\vx), \ldots h_n(\vx))$ and a unique polynomial $f_P \in \F[\vx]$ such that $f_{sym} =f_P(p_1(\vx),p_2(\vx), \ldots p_n(\vx))$, where $h_i(\vx)$ is the homogeneous symmetric polynomial of degree $i$ and $p_i(\vx)$ is the power symmetric polynomial of degree $i$. 
\end{theorem} 

\subsection{Schur polynomials}
\label{sec:prelim-schur}
A partition of a natural number $d$ is any sequence $\lambda = (\lambda_1,\lambda_2 \ldots, \lambda_{\ell} )$ of non-negative integers in a non-increasing order $\lambda_1 \geq \lambda_2 \ldots \geq  \lambda_{\ell} \geq 0$ such that $\sum_{i = 1}^{\ell} \lambda_i = d$.\footnote{Usually the $\lambda_i$s are assumed to be positive integers as defined earlier. For the sake of notational convenience we allow trailing zeroes in the definition of $\lambda$.}
The number of non-zero parts of $\lambda$ is called the length of $\lambda$ and is denoted by $l(\lambda)$. The weight of $\lambda$, denoted by $\abs{\lambda}$ is defined to be the sum of each individual component, i.e. $|\lambda|=\lambda_1 +
\lambda_2+ \cdots + \lambda_{l(\lambda)}$. If $|\lambda|=d$, then we say that $\lambda$ is a partition of the number $d$ or alternatively a partition of \emph{degree} $d$.  \chandra{Again the definition of $\lambda$ seems confusing to me. Either we should make the last entry as $\lambda_n$ instead of $\lambda_l$ or change  $\lambda_{\ell} \geq 0$ to $\lambda_{\ell} > 0$  as defined in the introduction. I personally feel we should change it to $\lambda_n$  as the foot note justifies it.}
\nutan{Added a footnote to address this issue.}

Let $\lambda$ be a partition of the number $d$. A Ferrers diagram (or simply a diagram) of shape $\lambda$ is is a left-aligned two-dimensional array of boxes with the $i$th row containing $\lambda_i$ many boxes.  The conjugate of $\lambda$, denoted by $\lambda'$, is the diagram obtained by switching the rows and columns of the diagram of $\lambda$. 

\begin{definition}[Schur polynomials] \label{def: Schur defn}
Let $\lambda$ be a partition of degree $d$ and let $l(\lambda) \leq n$. Then the Schur polynomial $s_\lambda(\vx)$ is defined as 
	
	$$s_\lambda(\vx)= \frac{a_{\lambda+\delta}(\vx)}{a_\delta(\vx)}$$
	where $$\delta = (n-1,n-2, \ldots \ldots 2,1,0)$$ \quad  $$a_{\lambda+\delta}(\vx)=\det(x_i^{\lambda_j+n-j})_{1 \leq i,j \leq n}$$ \quad 
	$$a_{\delta}(\vx)=\det(x_i^{n-j})_{1 \leq i,j \leq n}= \prod\limits_{\substack{1 \leq i<j \leq n}}(x_i-x_j)$$
\end{definition}

We first observe that $s_\lambda(\vx)$ is a symmetric polynomial. To see this, note that if $x_i=x_j$ for any $i \neq j$ then $a_{\lambda+\delta}(\vx)$ is $0$. Thus, by \autoref{lem:remainder} and the fact that $x_i-x_j$ and $x_{i'}-x_{j'}$ do not share a common factor unless $\set{i, j} = \set{i', j'}$, $\prod_{i<j}(x_i-x_j)$ is factor of $a_{\lambda+\delta}(\vx)$ i.e., $a_\delta(\vx)$ is a factor of $a_{\lambda+\delta}(\vx)$. Therefore, $s_\lambda(\vx)$ is a polynomial. Moreover, for any permutation of variables, the sign changes in the numerator and the denominator are the same, and thus their ratio does not see a change in sign. This implies that $s_{\lambda}(\vx)$ is a symmetric polynomial. 

We now state the classical Jacobi-Trudi identities which relates Schur polynomials to the elementary symmetric and homogeneous symmetric polynomials.

\begin{theorem}[Jacobi-Trudi identities] \label {thm:jacobi-trudi}  \leavevmode
	
	\begin{enumerate}
		\item[(1)] $s_\lambda(\vx)=\det(h_{\lambda_i -i+j}(\vx))_{1\leq i,j \leq \ell}$, where $\lambda= (\lambda_1, \ldots, \lambda_\ell)$.
		
		\item[(2)] $s_\lambda(\vx)=$ $\det(e_{\lambda'_i -i+j}(\vx))_{1\leq i,j \leq m}$,  where $\lambda'$ is the conjugate of $\lambda$ and $m= l(\lambda')$ . 
	\end{enumerate} 
	In particular, 
	\vspace{0.3cm}
\begin{align*}
s_\lambda &=  \begin{vmatrix}
		h_{\lambda_1} & ~~~~~~~~~h_{\lambda_1+1}~~~~~~ ~~\ldots~~h_{\lambda_1+\ell-1}\\\\
		h_{\lambda_2-1} & ~~~~~~~~h_{\lambda_2} ~~~~~~~~~~~\ldots~~ h_{\lambda_2+\ell-2}\\ \\
		\vdots & ~~~~\vdots ~~~~~~~~\ddots ~~~~~~~~~~~~\vdots &  \\
		\\
		h_{\lambda_\ell-\ell+1} & ~~~~h_{\lambda_\ell-\ell+2} ~~~\ldots~~~~~~~h_{\lambda_\ell}\\
	\end{vmatrix}_{\ell \times \ell} \\
	&= \begin{vmatrix}
		e_{\lambda'_1} & ~~~~~~~~~e_{\lambda'_1+1}~~~~~~ ~~\ldots~~e_{\lambda'_1+m-1}\\\\
		e_{\lambda'_2-1} & ~~~~~~~~e_{\lambda'_2} ~~~~~~~~~~~\ldots~~ e_{\lambda'_2+m-2}\\ \\
		\vdots & ~~~~\vdots ~~~~~~~~\ddots ~~~~~~~~~~~~\vdots &  \\
		\\
		e_{\lambda'_m-m+1} & ~~~~e_{\lambda'_m-m+2} ~~~\ldots~~~~~~~e_{\lambda'_m}\\
	\end{vmatrix}_{m\times m}
\end{align*}
Note that the identity depends only on the properties of $\lambda$ (i.e. $\ell$ or $m$) and does not depend on the number of variables $n$. 
\end{theorem}

\begin{theorem}
For any $\lambda$, $s_\lambda(\vx)$ can be computed using a small ABP, hence by a small algebraic circuit.
\end{theorem} 

\begin{proof}
For polynomials $P,Q \in \mathbb{C}[x_1,x_2 \ldots x_n]$ such that both $P$ and $Q$ have small ABPs, then by \autoref{lem:division of formulas} $R=\frac{P}{Q}$ also has a small ABP. Homogenization or interpolation can be used to extract the required polynomial without much blow up. Here both $a_{\lambda+\delta}(\vx)$, $a_\delta(\vx)$ have small ABPs (as they are small determinants), thus $s_\lambda(\vx)$ has an ABP of polynomial size which also implies that is has an algebraic circuit of polynomial size.
\end{proof}

It is well-known that $a_\delta(\vx)$, also known as the Principal Vandermonde Determinant, has a small algebraic formula. 
However much less is known about the complexity of $a_{\lambda+\delta}(\vx)$. These polynomials are also known as Generalized Vandermonde determinants and are well studied (see for instance~\cite{Heineman1929}). 
To the best of our knowledge, before this work it was not known whether for all $\lambda$, $a_{\lambda+\delta}(\vx)$s  have small formulas. Suppose they did, then by \autoref{lem:division of formulas}, we get that $s_\lambda(\vx)$ also have small formulas for all $\lambda$. 
In this paper we show that there exists some $\lambda$ for which $s_\lambda(\vx)$ does not have a small formula unless the Determinant has a small formula. This in particular implies that there exist $\lambda$ such that  $a_{\lambda+\delta}(\vx)$ cannot be computed using small formulas (unless the Determinant can be computed by a small formula).

\section{Proofs of main results}
\subsection{Proof of Lemma~\ref{lem:key-intro}}\label{sec:key lemma}
We start with the following definition, which is crucial for our proofs. 
\begin{definition}[Property S]\label{def:prop s}
A set of $n$ variate polynomials $\{q_1, q_2, \ldots, q_k\} \subseteq \F[\vx]$ is said to satisfy Property S, if there exists an $\va \in \F^n$ such that
\begin{itemize}
\item For all $i \in [k]$, $q_i(\va) = 0$, and,
\item The rank of the Jacobian matrix of $q_1, q_2, \ldots, q_k$ when evaluated at $\va$ is equal to its symbolic rank, i.e. ${\rank}_{\F}\left({\cal J}(q_1, q_2, \ldots, q_k)(\va)\right) = \rank_{\F(\vx)}\left({\cal J}(q_1, q_2, \ldots, q_k)\right)$.
\end{itemize} 
\end{definition}
Property S gives us a concrete way to capture an appropriate notion of \emph{niceness} of a set of algebraically independent polynomials. The following lemma which uses this notion is a key technical ingredient of our proofs. 
\begin{lemma}\label{lem:key}
Let $\{q_1, q_2, \ldots, q_k\} \in \F[x_1, x_2, \ldots, x_n]$ be a set of algebraically independent polynomials which satisfy Property S. Let $g \in \F[z_1, z_2, \ldots, z_k]$ be a homogeneous $k$ variate polynomial of degree equal to $d$ such that the composed polynomial $g(q_1, q_2, \ldots, q_k) \in \F[\vx]$ has an algebraic formula of size $s$ and depth $\Delta$. Then, $g(z_1, z_2, \ldots, z_k)$ can be computed by an algebraic formula of size $O(s^2n)$ and depth $\Delta+O(1)$.  
\end{lemma}

\begin{proof}
Let $\Phi$ be the formula of size $s$ which computes the polynomial $g\left(q_1(\vx), q_2(\vx), \ldots, q_k(\vx)\right)$. Since $q_1, q_2, \ldots, q_n$ satisfy Property S, there is an $\va \in \F^n$ such that $q_1(\va) = q_2(\va) = \cdots = q_n(\va) = 0$ and ${\rank}_{\F}\left({\cal J}(q_1, q_2, \ldots, q_k)(\va)\right) = \rank_{\F(\vx)}\left({\cal J}(q_1, q_2, \ldots, q_k)\right)$. Moreover, since they are algebraically independent, the rank of $\left({\cal J}(q_1, q_2, \ldots, q_k)(\va)\right)$ is equal to $k$. Thus, $${\rank}_{\F}\left({\cal J}(q_1, \ldots, q_k)(\va)\right) = k \, .$$
Applying~\autoref{cor:taylor truncation} to each $q_i(\vx)$ around this point $\va \in \F^n$, we get  
\begin{align*}
q_i(\va + \vx) &= \sum_{j = 1}^n x_j\cdot \frac{\partial q_i}{\partial x_j} (\va)  \mod \langle \vx \rangle^2 \, .
\end{align*}
Observe that $i^{th}$ row of the matrix $\left({\cal J}(q_1, q_2, \ldots, q_k)(\va)\right)$ is the vector $\left(\frac{\partial q_i}{\partial x_1} (\va), \frac{\partial q_i}{\partial x_2} (\va), \cdots, \frac{\partial q_i}{\partial x_n} (\va) \right)$ and by the choice of $\va$, these vectors  are linearly independent. Thus, the homogeneous linear forms $u_1(\vx), u_2(\vx), \ldots, u_k(\vx)$ are linearly independent, where $u_i(\vx)$ is defined as 
\[
u_i(\vx) = \sum_{j = 1}^n x_j\cdot \frac{\partial q_i}{\partial x_j} (\va)\, .
\]
The rest of the proof follows immediately from the following two claims. We state the claims and use them to complete the proof of this lemma, and then  move on to prove the claims. 
\begin{claim}\label{claim:using homogeneity of g}
Let $d$ be the degree of $g(\vx)$. Then, the homogeneous component of degree $d$ of the polynomial $g(q_1, q_2, \ldots, q_k)$ is equal to $g(u_1, u_2, \ldots, u_k)$. 
\end{claim}

\begin{claim}\label{claim:invertible linear transformation}
If $g(u_1, u_2, \ldots, u_k)$ has a formula of size $s'$ and depth $\Delta$, then $g(\vz)$ has a formula of size at most $s'n$ and depth $\Delta+O(1)$. 
\end{claim}
To complete the proof of the lemma, observe that given the formula $\Phi$ of size at most $s$ and depth at most $\Delta$ which computes $g(q_1, q_2, \ldots, q_k)$, we know from~\autoref{lem:homogeneous components} that the homogeneous component of degree $d$ of  $g(q_1, q_2, \ldots, q_k)$ can be computed by a formula $\Phi_1$ of size at most $O(s^2)$ and depth at most $\Delta+O(1)$. Moreover, from the homogeneity of $g$ and \autoref{claim:using homogeneity of g}, we also know that $\Phi_1$ computes the polynomial $g(u_1, u_2, \ldots, u_k)$, where $u_1, u_2, \ldots, u_k$ are linearly independent linear forms. Thus, from \autoref{claim:invertible linear transformation}, this implies that $g(z_1, z_2, \ldots, z_k)$ has a formula of size at most $O(s^2n)$ and depth $\Delta+O(1)$. This completes the proof of the lemma, modulo the two claims which we prove next. 
\end{proof}

\begin{proof}[Proof of \autoref{claim:using homogeneity of g}]
Let $f_1, f_2, \ldots, f_k \in \F[\vx]$ be polynomials which are zero modulo $\langle \vx \rangle^2$ (i.e. every monomial in $f_1, f_2, \ldots, f_k$ has degree at least $2$) such that for every $i$, $q_i(\va + \vx) = u_i(\vx) + f_i(\vx)$. Since $g$ is a  homogeneous polynomial of degree $d$, it can be expressed as 
\[
g(\vz) = \sum_{\vb \in \N^k, |\vb|_1 = d} \alpha_{\vb} \vz^{\vb} \,  ,
\] for field constants $\alpha_{\vb}$. 
  \nutan{Change all occurrences of $\ve$ to $\vb$?}
Let $\vb \in \N^k$ be any vector such that $|\vb|_1 = d$. Observe that the homogeneous component of degree  $d$ of the polynomial 
\[
\prod_{j = 1}^k q_j(\va + \vx)^{b_j} = \prod_{j = 1}^k (u_j(\vx) + f_j(\vx))^{b_j}
\] is equal to $\prod_{j = 1}^k u_j^{b_j} \, .
$
  Thus, by linearity, the homogeneous component of degree $d$ of  \[
  g\left(q_1(\va + \vx), q_2(\va + \vx), \ldots, q_k(\va + \vx)\right) = \sum_{\vb \in \N^k, |\vb|_1 = d} \alpha_{\vb} \cdot \prod_{j = 1}^k (q_j(\va + \vx ))^{\vb_j}\]
   equals \[
   \sum_{\vb \in \N^k, |\vb|_1 = d} \alpha_{\vb}\cdot  \prod_{j = 1}^k (u_j)^{\vb_j} \, , 
   \] which, in turn is the equal to the polynomial $g(u_1, u_2, \ldots, u_k)$.
\end{proof} 
\begin{proof}[Proof of \autoref{claim:invertible linear transformation}]
The idea for the proof of this claim is to show that each variable $x_j$ can be replaced by a homogeneous linear form $\ell_j(\vz)$ in the variables $\vz$ such that for every $i \in [k]$, the linear form $u_i$ satisfies $u_i(\ell_1(\vz), \ell_2(\vz), \ldots, \ell_n(\vz)) = z_i$. Thus, under this linear transformation, the composed polynomial $g(u_1(\vx), u_2(\vx), \ldots, u_k(\vx)) \in \F[\vx]$ gets mapped to the polynomial $g(z_1, z_2, \ldots, z_k)$.  Once we can show an existence of these linear forms $\ell_1, \ell_2, \ldots, \ell_n$, the bounds on the formula size and depth follow immediately since all we need to do to obtain a formula for $g(\vz)$ is to replace every occurrence of a variable in $\vx$, e.g $x_j$ by the linear form $\ell_j(\vz)$. Since every such linear form has a formula of size at most $k$ and depth $O(1)$, this process blows up the formula size by a multiplicative factor of at most $O(k)$ and the depth by an additive factor of $O(1)$. 

Intuitively, to obtain these linear forms, we just \emph{solve} the system of linear equations  $U \cdot \vx^{T} = \vz^T$, where $U$ is the $k \times n$ matrix whose $i^{th}$ row is $u_i$.  Since  the rank of $U$ is equal to $k$,  let $U'$ be an invertible $k \times k$ submatrix of $U$, and let $V$ be the inverse of $U'$ and let $v_1, v_2, \ldots, v_k$ be the rows of $V$. Moreover, for brevity, let us assume that $U'$ consists of the first $k$ columns of $U$. Observe that $(V\cdot U)$ is a $k \times n$ matrix such that its leftmost $k \times k$ sub-matrix is the identity matrix. We are now ready to define the linear forms $\{\ell_j : j \in [n]\}$.  For $j \in [k]$, let $\ell_j(\vz)$ be equal to $v_j(\vz)$ and for $j > k$, $\ell_j(\vz)$ is defined to be zero. It is straightforward to check that this definition satisfies the desired property and we skip the details. 
\end{proof}
\subsection{Formula complexity of Schur polynomials}
\label{sec:schur-formula-complexity}
We are now ready to prove \autoref{thm:main-intro}. Our first stop, which we reach in the next two lemmas, is to show that which shows that the elementary symmetric polynomials of degree at most $n-1$ on $n$ variables are \emph{well behaved}. We start by establishing a sufficient condition for their Jacobian matrix to have full rank at a point. 
\begin{lemma}\label{lem:jacobian of esyms}
Let $\vx = (x_1, x_2, \ldots, x_n)$ be an $n$-tuple of variables. Let $J(\vx)$ be the Jacobian of $e_1(\vx), e_2(\vx), \ldots, e_{n-1}(\vx)$, and let $\vb = (b_1, b_2, \ldots, b_n) \in \F^n$ be such that for all $i \neq j$, $b_i \neq b_j$. Then, 
\[
\rank_{\F(\vx)}\left(J(\vx) \right) = \rank_{\F}(J(\vb)) = n-1 \, .
\]
\end{lemma}
\begin{proof}
Using \autoref{fact:sym ad}, we can observe that  the $n$-variate polynomials $e_1(\vx), e_2(\vx), \ldots, e_{n-1}(\vx)$ are algebraically independent.
Therefore, from \autoref{thm:jac and alg ind} we know that $\rank_{\F(\vx)}\left(J(\vx) \right)$ is equal to $n-1$. 

Let $J'(\vx)$ be any $n-1 \times n-1$ submatrix of $J(\vx)$ of rank equal to $n-1$ and by symmetry we can assume that the columns in $J'(\vx)$ come from the first $n-1$ columns of $J(\vx)$. Thus, for $i, j \in [n-1]$, the $(i, j)$ entry of $J'(\vx)$ is equal to $\frac{\partial e_i}{\partial x_j}$. We now show that the rank of $J'(\vb)$ over $\F$ is equal to $n-1$ and this would complete the proof. To this end, we observe that the determinant of $J'(\vx)$ is a non-zero scalar multiple of $\prod_{i, i' \in [n-1], i \neq i'}(x_i-x_{i'})$. Since the coordinates of $\vb$ are all distinct, this determinant remains non-zero on $\vb$. 

From the definition of $J'(\vx)$, we know that the entries in its $i^{th}$ row are homogeneous polynomials of degree equal to $i-1$. Thus, $\det(J'(\vx))$ is a homogeneous polynomial in $\vx$ of degree equal to \[
0 + 1 + 2 + \cdots + n-2 = \frac{(n-2)(n-1)}{2} \, . \] Recall that the $(i, j)$ entry of $J'(\vx)$, is equal to $\frac{\partial e_i}{\partial x_j}$, which is equal to $\sum_{S \subseteq [n]/ \{i\}}\prod_{k \in S} x_k$. Thus, if we replace every occurrence of the variable $x_i$ by the variable $x_{i'}$ for $i \neq i' \in [n-1]$, then columns $i, i'$ in $J'(\vx)$ become identical, and hence $\det(J'(\vx))$ is identically zero. Thus, by \autoref{lem:remainder}, $(x_i-x_{i'})$ is a factor of $\det(J'(\vx))$. Also, for any two distinct sets $\{i_1, i_1'\}$ and $\{i_2, i_2'\}$ where $i_1 \neq i_1'$ and $i_2 \neq i_2'$, the polynomials $(x_{i_1}-x_{i_1'})$ and $(x_{i_2}-x_{i_2'})$ do not share a non-trivial divisor. Thus, the determinant of $J'(\vx)$ must be divisible by the polynomial $\prod_{i, i' \in [n-1], i \neq i'}(x_i-x_{i'})$. Moreover, we observed that the degree of determinant of $J'(\vx)$ is equal to $\frac{(n-2)(n-1)}{2}$, which is also equal to the degree of $\prod_{i, i' \in [n-1], i \neq i'}(x_i-x_{i'})$. Thus, they must be non-zero scalar multiples of each other. This observation, together with the fact that the coordinates of $\vb$ are all distinct, shows that $\det(J'(\vx))$ is non-zero at $\vb$ and hence, $\rank(J(\vb))$ equals $n-1$ over $\F$.
\end{proof}
We now use \autoref{lem:key} to show that $e_1(\vx), e_2(\vx), \ldots, e_{n-1}(\vx)$ satisfy Property S, where $n$ is the number of variables.
\begin{lemma}\label{lem:esym prop s}
Let $\vx = (x_1, x_2, \ldots, x_n)$ be an $n$-tuple of variables. Then, the set of elementary symmetric polynomials of degree at most $n-1$ on $\vx$, i.e. the set $\{e_1(\vx), e_2(\vx), \ldots, e_{n-1}(\vx)\}$ of polynomials satisfies Property S. 
\end{lemma}

\begin{proof}
Let $\va = (1, \omega, \omega^2, \ldots, \omega^{n-1})$, where $\omega$ is the primitive $n^{th}$ root of unity. So, we have the following identity 
\[
z^n-1 = \prod_{i=1}^{n} (z-\omega^{i-1}) \, .
 \] 
However, from ~\autoref{fact:esym and roots}, we also know that $z^n - 1$, which equals $\prod_{i = 1}^n (z-a_i)$ can be expressed as
\[
\prod_{i = 1}^n (z-a_i) = z^{n} - c_1\cdot z^{n-1} + c_2\cdot z^{n-2}- \cdots + (-1)^{n} c_n \, ,
\]
where $c_i = e_i(\va)$.  
Comparing these two expressions for $z^n -1$, we get that for all $1\leq i < n$, $c_i = e_i(\va)=0$ . Thus, $\va=(1, \omega, \omega^2, \ldots, \omega^{n-1})$ is a common zero of $e_1, e_2, \ldots, e_{n-1}$, which satisfies the first item in \autoref{def:prop s}. Moreover, since $\omega$ is a primitive $n^{th}$ root of one, we also know that for all pairs $i \neq j$, $a_i \neq a_j$. Therefore, by \autoref{lem:jacobian of esyms}, $\va$ satisfies the second condition in \autoref{def:prop s}. Thus, the $n$ variate polynomials $e_1(\vx), e_2(\vx), \ldots, e_{n-1}(\vx)$ satisfy Property S.  
\end{proof}

%
We now observe that analogous statements are also true for homogeneous symmetric polynomials and power symmetric polynomials as well.  
\begin{lemma}\label{lem:hsym psym prop S} 
Let $\vx = (x_1, x_2, \ldots, x_n)$ be an $n$-tuple variables. Then the set of complete symmetric polynomials of degree at most $n-1$ on $\vx$, i.e. the set $\{h_1(\vx), h_2(\vx), \ldots, h_{n-1}(\vx)\}$ of polynomials satisfies Property S. Similarly, the set $\{p_1(\vx), p_2(\vx), \ldots, p_{n-1}(\vx)\}$ of power symmetric polynomials of degree at most $n-1$ also satisfies Property S.
\end{lemma}

\begin{proof}[Proof sketch]
	The proof goes via  generating functions of $e_i$, $h_i$, $p_i$ and the relations among them. We denote the generating function of $e_i$, $h_i$, $p_i$ by $E(t),H(t),P(t)$ respectively. The following relations are known between these polynomials. (See for instance~\cite{}.)
	$$E(t)= \prod_{i \geq 1} (1+x_it)=\sum_{k \geq 0} e_kt^k$$ 
	$$H(t)=\prod_{i \geq 1} \frac{1}{1-x_it}=\sum_{k \geq 0}  h_kt^k$$
	\quad $$ \quad P(t)=\prod_{i\geq 1}\frac{x_i}{1-x_it}= \sum_{k \geq 1}  p_kt^{k-1}$$
It is easy to see that $E(-t)H(t)=1$.  Therefore we get, 
\begin{equation} \label{equ: h to e gen fun}
H(t)= \frac{1}{E(-t)} =\frac{1}{1-e_1t+e_2t^2 \ldots +(-1)^ne_nt^n}
\end{equation}
From \autoref{lem:esym prop s}, we know there exists a point $\va=(a_1, a_2, \ldots, a_n)$, where $e_1(\va), e_2(\va), \ldots, e_{n-1}(\va)$ vanish and $e_n(\va)$ is non-zero. We evaluate the above equation at the same $\va$. We denote the above evaluation by $H(t)|_{\vx=\va}$.
$$H(t)|_{\vx=\va}=1- (-1)^ne_n(\va)t^n +((-1)^ne_n(\va)t^n)^2 \ldots $$ 
	
Observing the equation for $H(t)|_{\vx=\va}$, it follows that at point $\va$ , $h_1(\va), h_2(\va), \ldots, h_{n-1}(\va)$ are zero and $h_n(\va)$ is non-zero. 
	
An analogous relation between $P(t)$ and $E(t)$, given by
	  \[
	  P(t)=\frac{E^\prime(-t)}{E(-t)} \, ,
	  \] 
where $E^\prime(-t)$ is the first derivative of $E(-t)$ with respect to $t$, can be used to show that the power symmetric polynomials $p_1(\va),\ldots, p_{n-1}(\va)$ are zero and  $p_n(\va)$ is non-zero. 
Thus, we are halfway towards showing that the set of power symmetric polynomials of degree at most $n-1$ and homogeneous symmetric polynomials of degree at most $n-1$ also satisfy Property S. It remains to be argued that the rank of the Jacobian of these polynomials at $\va$ is equal to $n-1$. The proof of this closely follows the analogous argument for the elementary symmetric polynomials, as in the proof of \autoref{lem:esym prop s}. We skip the remaining details. 
\end{proof}
From \autoref{def:prop s} of Property S and \autoref{thm:jac and alg ind}, it follows that is a set $A$ of polynomials satisfies Property S, then all the non-empty subsets of $A$ also satisfy Property S. Thus, we have the following corollary of \autoref{lem:esym prop s} and \autoref{lem:hsym psym prop S}.
\begin{corollary}\label{cor:prop s subsets}
Let $i_1 < \dots < i_k < n$ be positive integers. Consider the sets $E = \{e_{i_1}(\vx), \dots, e_{i_k}(\vx)\}$ and $H = \{h_{i_1}(\vx), \dots, h_{i_k}(\vx)\}$ of elementary and homogenous symmetric polynomials respectively in $n > i_k$ variables. Then both $E$ and $H$ satisfy property $S$. 
\end{corollary}

We are now ready to prove the main theorem. 

\begin{theorem}[Main theorem]
\label{thm:main-intro proof}
Let $\lambda= (\lambda_1,\ldots,\lambda_\ell)$ be a partition of $d$ such that  $\lambda_i \geq  \lambda_{i+1}+(\ell-1)$ for all $i \in [\ell-1]$, and $\lambda_{\ell} \geq \ell$. Then, for all $n$, such that $n \geq \lambda_1 + \ell$, if $s_\lambda(x_1,\ldots,x_n)$ has an algebraic formula of size $s$ and depth $\Delta$, then the $\ell\times \ell$ determinant $(\mathrm{det}_\ell )$ has an algebraic formula of size $\poly(s)$ and depth $\Delta + O(1).$  
\end{theorem}

\begin{proof}
	We use Jacobi-Trudi Identity from \autoref{thm:jacobi-trudi}, which expresses $s_\lambda(\vx)$ in the form of homogeneous symmetric determinant. 
\begin{align*}
s_\lambda &=  \begin{vmatrix}
		h_{\lambda_1} & ~~~~~~~~~h_{\lambda_1+1}~~~~~~ ~~\ldots~~h_{\lambda_1+\ell-1}\\\\
		h_{\lambda_2-1} & ~~~~~~~~h_{\lambda_2} ~~~~~~~~~~~\ldots~~ h_{\lambda_2+\ell-2}\\ \\
		\vdots & ~~~~\vdots ~~~~~~~~\ddots ~~~~~~~~~~~~\vdots &  \\
		\\
		h_{\lambda_{\ell}-\ell+1} & ~~~~h_{\lambda_{\ell}-\ell+2} ~~~\ldots~~~~~~~h_{\lambda_{\ell}}\\
		\end{vmatrix}
\end{align*}	
Let $M_{\lambda}$ denote the matrix on the right hand side in the above equality. By our choice of $\lambda$, observe that the highest degree entry of $M_{\lambda}$ is $h_{\lambda_1 + \ell -1}$, which has degree at most $n-1$, and the lowest degree entry of $M_{\lambda}$ is $h_{\lambda_{\ell}-\ell + 1}$ which for $\lambda_{\ell} \geq \ell$ has degree at least $1$. Moreover, from the choice of $\lambda$, it also follows that all the entries of $M_{\lambda}$ are distinct.

From \autoref{lem:hsym psym prop S}, we know that there exists a point $\va$ for which the $n$-variate polynomials $\{h_1({\vx}),h_2(\vx), \ldots h_{n-1}(\vx)\} $ satisfy the Property $S$, where $n$ can be taken to be strictly greater than $\ell^2$. Thus, now if we take the $\ell^2$-variate homogeneous polynomial $g$ to be the symbolic determinant of an $\ell \times \ell$ matrix, then $s_{\lambda}$ is obtained by a composition of $g$ with a subset of polynomials $h_1, h_2, \ldots, h_{n-1}$. 
 
Thus, by \autoref{lem:key}, we get that if $g(h_1({\vx}),h_2(\vx), \ldots h_{\ell^2}(\vx))$ (i.e $s_\lambda$) has an algebraic formula of size $s$ and depth $\Delta$, then $g(x_1,x_2, \ldots x_{\ell^2})$ also has a small formula of size $\poly(s)$ and depth $\Delta+O(1)$ . 
\end{proof} 

\begin{remark}[Contrasting hard and easy $\lambda$s]
\label{rem:easy-hard-lambda}
Let $\lambda = (\ell^2, \ell^2-\ell, \ldots, 2\ell, \ell, 0, 0, \ldots, 0,0,0)$ and let $\tilde{\lambda} = ( n\ell, (n-1)\ell, \ldots, \ell^2,\ell^2-\ell, \ldots, 3\ell, 2\ell, \ell)$, where we will take $n \geq \ell^2+\ell$ and $\ell>0$. It is easy to see that $\tilde{\lambda}$ forms an arithmetic progression in which the difference between the successive terms is $\ell$. Whereas $\lambda$ is a truncated arithmetic progression, in which  $n-\ell$-many elements are zeroes and the non-zero elements form an arithmetic progression. 

Here the structure of $\lambda$ and $\tilde{\lambda}$ is quite similar, but the algebraic complexities of $s_\lambda$ and $s_{\tilde{\lambda}}$ are  different. 
In particular, \autoref{thm:main-intro proof} is applicable to $\lambda$. Therefore we can conclude that $s_\lambda$ does not have a small algebraic formula unless the determinant has a small algebraic formula. 

On the other hand, there are small formulas for $s_{\tilde{\lambda}}$. To see this, observe that 
$\tilde{\lambda}+\delta= ((n-1)(\ell+1)+\ell, (n-2)(\ell+1)+\ell, \ldots, (\ell+1)+\ell, \ell)$, 
which is also an arithmetic progression in which the difference between successive terms is $\ell+1$. Therefore, we have 
\vspace{0.4cm}
$$
a_{\tilde{\lambda}+\delta}=
\begin{vmatrix}

x_1^{\ell} &~~x_2^{\ell} ~~~~~~~~~~~~ \ldots~~~~~~~~~~~~ ~ x_n^{\ell}\\  

~~~~~x_1^{(\ell+1)+\ell} &~~~~~~~x_2^{(\ell+1)+\ell}~~~~~~ \ldots~ ~~~~~~ ~~~~~x_n^{(\ell+1)+\ell}\\

\vdots & ~~~\vdots ~~~~~~~~~~~~\ddots ~~~~~~~~~~~~~\vdots &  \\

~~~~~~~~x_1^{(n-1)(\ell+1)+\ell} & ~~~~~~~~~ x_2^{(n-1)(\ell+1)+\ell} ~~~\ldots ~~~~~~~~~~x_n^{(n-1)(\ell+1)+\ell}\\

\end{vmatrix}_{n\times n}
$$

\vspace{0.4cm}
$$
a_{\tilde{\lambda}+\delta}= \bigg{(} \prod_{i=1}^{n} x_i^{\ell}\bigg{)}
\begin{vmatrix}
1 & 1 & \ldots &1 \\
x_1^{\ell+1} & x_2^{\ell+1} & \ldots & x_n^{\ell+1}\\
\vdots & \vdots &\ddots & \vdots &  \\

x_1^{(n-1)(\ell+1)} & x_2^{(n-1)(\ell+1)} & \ldots & x_n^{(n-1)(\ell+1)}\\

\end{vmatrix}_{n\times n}
$$  

\begin{align*}
s_{\tilde{\lambda}} &= \frac{\bigg{(} \prod_{i=1}^{n} x_i^{\ell}\bigg{)} \prod\limits_{\substack{i<j}} (x_j^{\ell+1}-x_i^{\ell+1})}{\prod\limits_{\substack{i<j}} (x_j-x_i)}\\
\end{align*}

In the above expression the numerator and denominator have small formulas and therefore using \autoref{lem:division of formulas} we get that $s_{\tilde{\lambda}}$ has a small formula of size $\poly(\ell,n)$.
\end{remark}

\subsubsection{Generalization to Skew Schur Polynomials}

A straightforward generalization of the previous result is to prove that a class of skew Schur polynomials is also hard for formulas. They can be defined via a Jacobi-Trudi like formula. 

\begin{theorem}\label{def: skew schur definition}
Let $\mu$ and $\lambda$ be partitions with $\mu_i \leq \lambda_i$ for every part $i$. Then, the skew Schur polynomial $s_{\lambda/\mu}$ satisfies

\begin{enumerate}
\item $s_{\lambda/\mu} = \det(h_{\lambda_i - \mu_j - i + j})_{1\leq i, j \leq k}$ where $k = l(\lambda).$
\item $s_{\lambda/\mu} = \det(e_{\lambda_i' - \mu_j' - i+j})_{1 \leq i, j \leq k}$ where $k = l(\lambda')$.  
\end{enumerate}
\end{theorem}

From these definitions of skew Schur polynomials, we can see that they also have ABPs of polynomial size, like Schur polynomials. However, skew Schur polynomials are in general linear combinations of Schur polynomials, by the Littlewood-Richardson rule $$s_{\lambda / \mu} = \sum_{\nu \vdash n - m} c^\lambda_{\mu, \nu} s_\nu.$$ The Littlewood-Richardson coefficients $c^\lambda_{\mu, \nu}$ count tableau whose Young diagram is of shape $\lambda / \mu$ and whose content satisfy certain technical conditions. They are also important in representation theory as they describe how Schur polynomials multiply, or equivalently how a tensor product of polynomial GL(n) representations decomposes into irreducible representations. They are also known to be \#P-hard to compute \cite{narayanan2006complexity}. 

Hardness of skew Schur polynomials assuming hardness of determinant follows from the following lemma, \autoref{cor:prop s subsets} and \autoref{lem:key}.
\begin{lemma}\label{lem: skew distinct}
Let $l \geq 2$ and $\mu_1 \geq 1$ be positive integers. Let $\lambda, \mu$ be partitions with $l$ parts with $\lambda_i = (l-(i-1))l + \mu_1$ and $\mu_i = \mu_1$ for all $i < l$ and $\mu_l = \mu_1 - 1$. Then all entries of the homogeneous Jacobi-Trudi determinant $s_{\lambda/\mu}$ are distinct. 
\end{lemma}

\begin{proof}
For simplicity, we call the integer $k$ the label of the homogeneous symmetric polynomial $h_k$. Here, the $(i,j)$ entry of the matrix is $h_{(l-(i-1))l - i + j}$ for $1 \leq j < l$ and $h_{(l-(i-1))l + 1 - i + j}$ for $j = l$. Hence, for a fixed row $i$, all entries are distinct and the labels increase from left to right. Furthermore,

$$(l-i)l + 1 - (i+1) + l = (l-i)(l-1) + 1 - i < l(l-(i-1)) - i + 1$$

so the label of the last entry of row $(i+1)$ is strictly less than the label of the first entry of row $i$. Hence, all entries of the Jacobi-Trudi determinant are distinct.
\end{proof}

\subsection{Extensions of the results of Bl{\"a}ser and Jindal \cite{blser_et_al:LIPIcs:2018:10140}}\label{sec:BJ}
\subsubsection*{Shifted variants for formulas}
A fairly direct consequence of our techniques is the following theorem which can be considered a partial generalization of the result of Bl{\"a}ser-Jindal \cite{blser_et_al:LIPIcs:2018:10140} for algebraic formulas. 

\begin{theorem}\label{thm: bj shifted}
There exist field constants $a_1, a_2, \ldots, a_n$ such that the following is true:  if for an $n$-variate homogeneous polynomial $g$ over $\F$, the composition $g(e_1 - a_1, \ldots, e_n - a_n)$ has a formula of size at most $s$, then the polynomial $g(\vy)$ has a formula of size at most $O(s^2n)$.
	
\end{theorem}

\begin{proof}
	We will primarily use \autoref{lem:key} to prove the above theorem. Note that the set $\{e_1-a_1, \ldots , e_n-a_n\}$ continue to be algebraically independent as the Jacobian matrix ${\cal J}(e_1-a_1, \ldots, e_n-a_n)={\cal J}(e_1, \ldots, e_n)$. Recall that $(e_1 \ldots e_n)$ are algebraically independent over $n$ variables and thus has a full rank Jacobian matrix, which also implies the algebraic independence of $(e_1-a_1, \ldots, e_n-a_n)$ using \autoref{thm:jac and alg ind}. The next step would be to find a point $\vc$ where $e_i(c)-a_i$ is zero for all $i$. We also need to make sure that at point $\vc$ the ${\cal J}(e_1-a_1, \ldots, e_n-a_n)$ matrix has full rank. It is easy to verify that the determinant of ${\cal J}(e_1-a_1, \ldots, e_n-a_n)$ is a polynomial of degree $n \choose 2$. As per our assumption, the field we use is quite large, in fact much larger than $n^2$. Thus, from the Polynomial Identity lemma (\autoref{lem:polynomial identity lemma}), we know that over every large enough set $S \subseteq \F$, there exists a point $\vc \in S^n$  at which the  determinant of ${\cal J}(e_1-a_1, \ldots, e_n-a_n)$ is non-zero and hence, the matrix ${\cal J}(e_1-a_1, \ldots, e_n-a_n)(\vc)$ is full rank.  By setting $a_i=e_i(\vc)$ for all $i$, $\{e_1-a_1, \ldots , e_n-a_n\}$ satisfy the Property $S$ mentioned in \autoref{def:prop s} for point $\vc$. The proof of this theorem is now an immediate corollary of \autoref{lem:key}. 
\end{proof}
The proof above gives a slightly stronger statement than what is stated in \autoref{thm: bj shifted}. More precisely, the statement of \autoref{thm: bj shifted} holds for many many $a_1, a_2, \ldots, a_n$. To see this, note that the only property that the proof above uses about $a_1, a_2, \ldots, a_n$ is that $\va = (a_1, a_2, \ldots, a_n)$ is a point in the image of the polynomial map $\sigma$ from $\F^n$ to $\F^n$ which is given by mapping $\vy \in \F^n$ to $\left(e_1(\vy), e_2(\vy), \ldots, e_n(\vy) \right)$ such that there is a point $\vc$ in the pre-image of $\va$ where the Jacobian ${\cal J}(e_1-a_1, \ldots, e_n-a_n)(\vb)$ (which is equal to ${\cal J}(e_1, \ldots, e_n)(\vb)$) is full rank. Now, observe that the map $\sigma$ is invertible. To see this, note that $\sigma$ can be thought of as mapping $n$  roots $b_1, b_2, \ldots, b_n$ of the univariate polynomial $\prod_{i = 1}^n (z-b_i)$ to its coefficients, and hence its inverse is the map which maps the $n$ coefficients of a monic degree $n$ polynomial to its roots. 

Thus, by \autoref{lem:polynomial identity lemma} if we take $\vb$ to be a random point from a large enough grid, then the Jacobian ${\cal J}(e_1-a_1, \ldots, e_n-a_n)$ has rank $n$ with high probability. Moreover, whenever this event happens, \autoref{thm: bj shifted} holds with $\va$ being set to be the image of $\vb$ under $\sigma$. 

\subsubsection*{Generalizing the results in \cite{blser_et_al:LIPIcs:2018:10140} to other bases}
We know that the fundamental theorem of symmetric polynomials also holds for other bases and not just for elementary symmetric basis.  For any given $n$-variate symmetric polynomial $f_{sym}$, Bl{\"a}ser-Jindal  efficiently finds $f$ such that  $f_{sym}=f(e_1,e_2 \ldots e_n)$. The degree of $f$ is also given apriori. We generalize the Bl{\"a}ser-Jindal work to other bases such as homogeneous symmetric base and power symmetric base efficiently. In order to prove that, we need to show there exists an efficient transformation which can represent any elementary symmetric polynomial in the form of homogeneous symmetric or power symmetric polynomial. The following lemma illustrates that.

\begin{lemma}\label{e_k to h_k formula complexity} Any $n$-variate elementary symmetric polynomial of degree $k$ can be written as an algebraic combination of homogeneous symmetric polynomials( or power symmetric polynomials) using a small formula.
	
	\begin{proof}
		It is well known that these transformations are doable using  ABP of polynomial size as the transformation uses small determinants \cite{MACD} . We prove the same for formula.
		Recall that $E(t)H(-t)=1$
		$$E(t)= \frac{1}{H(-t)}=\frac{1}{1-h_1t+h_2t^2 \ldots +(-1)^nh_nt^n}= \frac{1}{1-z}= \sum_{i \geq 0 }z^i$$
		where $z= h_1t-h_2t^2 \ldots +(-1)^{n-1}h_nt^n$ \\
		
		Consider the truncated polynomial $A(t)= E(t) \mod {\langle z\rangle} ^{k+1}$, where $\langle z\rangle$ is the ideal generated by $z$. Now we use interpolation to find the coefficient of $t^k$ in the polynomial $A(t)$, which is precisely $e_k$. $A(t)$ can have degree at most $nk$, which implies a formula size of $O(nk)$ for $A(t)$ (assuming the field to be algebraically closed). But, $k \leq n$, hence the trivial formula complexity for expressing $e_k$ in the form of $h_i$'s is $O(n^4)$ using \autoref{lem:homogeneous components}.
		
		From the definition of the generating functions $P(t)$ and $E(t)$, it is easy to verify that 
		
		$$E(t)=e^{\int P(-t) dt}= e^{\int (\sum_{m \geq 1} p_m t^{m-1}) dt}= e^{({\sum_{m \geq 1} \frac{p_mt^m}{m!}})} = 1+q+\frac{q^2}{2!} \ldots~, $$ 
		
		where $q={\sum_{m \geq 1} \frac{p_mt^m}{m!}}$ \vspace{0.2cm}
		
		Now we consider the truncated polynomial containing degree up to $q^k$ and interpolate this polynomial to get the coefficient of $t^k$. The formula complexity for expressing $e_k$ as power symmetric polynomials is $O(n^4)$. Also, the proof outline is very similar to the homogeneous symmetric case. 
		
	\end{proof}
\end{lemma}

\begin{theorem}[Bl{\"a}ser-Jindal for other bases] \label{Thm:Blaser-Jindal other base}For any $n$-variate symmetric polynomial $f_{sym} \in \mathbb{C}[\vx]$, we can efficiently compute the polynomials $f_E,f_H,f_P \in \mathbb{C}[\vx]$ that are $n$ variate and unique s.t $f_{sym}=f_E(e_1,e_2 \ldots e_n)=f_H(h_1,h_2 \ldots h_n)=f_P(p_1,p_2 \ldots p_n)$.
\end{theorem}

\begin{proof}
	Bl{\"a}ser-Jindal proves that $f_{E}$ can be efficiently computed using an algebraic circuit. They prove that the circuit size for computing $f_E$ is bounded by $O(d^2S(f_{sym})+$\poly$(n,d))$ where $d$ is the degree of $f_E$, $S(f_{sym})$ denotes the circuit size of $f_{sym}$ and $n$ is the number of working variables. We extend their work for computing $f_H$ and $f_P$.  To prove that, we shall use Bl{\"a}ser-Jindal method as a black box. We use Bl{\"a}ser-Jindal technique and get the circuit for $f_E$. We denote this circuit by $\mathcal{C}_{f_E}$. $\mathcal{C}_{f_E}$ can be visualized as a circuit having elementary symmetric polynomials($e_i$'s) as its input.  But, from \autoref{e_k to h_k formula complexity}, we know $e_i$'s can be uniquely expressed in the form of $h_i$'s and $p_i$'s using a  formula of polynomial size. We denote the modified circuit as  $\mathcal{C}_{f_H}$ after transforming $e_i$'s to $h_i$'s in $\mathcal{C}_{f_E}$. The output of $\mathcal{C}_{f_H}$  is $f_H$ because $h_i$'s are algebraically independent and also satisfy fundamental theorem of symmetric polynomials. The proof for $f_P$ is similar. Thus $f_H$ and $f_P$ can still be computed using a circuit of size $O(d^2S(f_{sym})+$\poly$(n,d))$ as the transformation just adds some poly factor to the previously calculated size.
\end{proof}

Bl{\"a}ser-Jindal method works for the circuit and we do not know whether it works for ABPs or formula. If it works for ABP, then the basis transformation is trivial using a small determinant (i.e a small ABP).  If it works for formula, then \autoref{Thm:Blaser-Jindal other base} would be useful to extend the notion for other bases while still staying in the  formula regime.

\subsection{Partial derivatives of a product of algebraically independent polynomials}\label{sec:partial derivatives}
We digress a little in this section to discuss another application of the ideas used in the proof of \autoref{lem:key} to a question of Amir Shpilka on the partial derivative complexity of a product of algebraically independent polynomials. We start with the definition of the partial derivative complexity.

\begin{definition}
 The partial derivative complexity of a polynomial $P \in \F[\vx]$ is the dimension of the linear space of polynomials over $\F$ spanned by all the partial derivatives of $P$.
\end{definition}
 The following question was asked by Amir Shpilka and our techniques provide a partial answer. As far as we are aware, the general question remains open. 
\begin{question}[Shpilka]\label{q:amir}
Let $g_1, g_2, \ldots, g_k \in \F[\vx]$ be algebraically independent polynomials. Then, prove (or disprove) that the partial derivative complexity of the the product $\prod_{i = 1}^k g_i(\vx)$ is at least $\exp(\Omega(k))$.
\end{question}
A canonical example of polynomials satisfying the hypothesis is when $g_i(\vx) = x_i$. Thus, the product polynomial $\prod_{i = 1}^k g_i(\vx)$ is equal to the monomial $x_1\cdot x_2 \cdots x_k$, and indeed the partial derivative complexity of this monomial is at least $2^k$, since for every $S\subseteq [k]$, the monomial $\prod_{i \in S} x_i$ is a partial derivative and these monomials are all linearly independent over $\F$ .  Thus, in general, we cannot hope for a better lower bound on the partial derivative complexity of such polynomials. Using our techniques, we observe the following two statements which answer special cases of this question. 
\begin{theorem}\label{thm:amir 1}
Let  $g_1, g_2, \ldots, g_k \in \F[\vx]$ be algebraically independent polynomials which satisfy Property S. Then, the partial derivative complexity of $\prod_{i = 1}^k g_i(\vx)$ is at least $2^k$.  
\end{theorem}

\begin{theorem}\label{thm:amir 2}
Let  $g_1, g_2, \ldots, g_k \in \F[\vx]$ be algebraically independent polynomials. Then, there are field constants $a_1, a_2, \ldots, a_n$ such that the partial derivative complexity of $\prod_{i = 1}^k \left( g_i(\vx) + a_i\right)$ is at least $2^k$.  

In fact, the lower bound holds for almost all choices of $a_1, a_2, \ldots, a_k$. 
\end{theorem}
The following observation essentially follows from the definition of Property S (\autoref{def:prop s}) and \autoref{thm:taylor}. The proof is also implicit in the proof of \autoref{lem:key}. 
\begin{observation}\label{obs:good shift}
Let $q_1(\vx), q_2(\vx), \ldots, q_k(\vx) \in \F[\vx]$ be algebraically independent polynomials which satisfy Property S. Then, there is an $\va \in \F^n$ such that the degree zero homogeneous component of the the polynomials  $q_1(\vx + \va), q_2(\vx + \va), \ldots, q_k(\vx + \va)$ are all zero, and their homogeneous components of degree one are all linearly independent. 
 \end{observation}
\autoref{thm:amir 1} and \autoref{thm:amir 2} are essentially immediate consequences of \autoref{obs:good shift} and some standard properties of partial derivatives, which we now discuss.  
 \begin{lemma}\label{lem: partial derivative props}
Let $P \in \F[\vx]$ be a polynomial. Then, 
\begin{enumerate}
\item For every $\va \in \F^n$, the partial derivative complexity of $P$ is equal to the partial derivative complexity of $P(\vx + \va)$. More generally, the partial derivative complexity is invariant under invertible linear transformation of variables. 
\item Let $i$ be the degree of the lowest degree homogeneous component of $P$ which is non-zero and let $P_i$ denote this homogeneous component. Then, the partial derivative complexity of $P$ is at least as large as the partial derivative complexity of $P_i$. 
\item The partial derivative complexity of a product of $k$ linearly independent homogeneous linear forms is at least $2^k$. 
\end{enumerate} 
\end{lemma}
We briefly sketch the main ideas in the proof. 
\begin{proof}[Proof Sketch]
For the first item, we prove the statement for first order partial derivatives. The argument easily extends to higher order derivatives as well. By the chain rule, $\frac{\partial P(\vx + \va)}{\partial x_i}$ equals $\frac{\partial P(\vx + \va)}{\partial (x_i + a_i)} \cdot \frac{\partial (x_i + a_i)}{\partial x_i}$. Observe that this is equal to the polynomial obtained by taking the partial derivative $\frac{\partial P(\vx)}{\partial x_i}$ of the original polynomial and then shifting the variables by $\va$, i.e. replacing $x_j$ by $x_j + a_j$ for every $j$. Thus, the linear space spanned by the first order partial derivatives of $P(\vx + \va)$ is equal to the linear space obtained by taking the space of first order partial derivatives of $P(\vx)$ and shifting the variables by $\va$. It is not hard to see that this preserves the dimension of the space. The proof of the \emph{moreover} part needs a bit more care, but follows similarly. 

For the second item, observe that for any set of polynomials $\{Q_1, Q_2, \ldots, Q_t\}$, the dimension of the linear span of $\{Q_1, Q_2, \ldots, Q_t\}$ is at least as large as the dimension of the linear span of the lowest degree non-zero homogeneous components of $Q_1, Q_2, \ldots, Q_t$. Also, observe that for any monomial $\vx^{\vb}$, if the partial derivative $\frac{\partial P_i}{\partial \vx^{\vb}}$ is non-zero, then the lowest degree non-zero homogeneous component of $\frac{\partial P}{\partial \vx^{\vb}}$ equals $\frac{\partial P_i}{\partial \vx^{\vb}}$. Now, let $S$ be the set of monomials such that the space of partial derivatives of $P_i$ with respect to monomials in $S$ is a basis for the linear space of all partial derivatives of $P_i$. From the two earlier observations in this paragraph, it follows that the derivatives of $P$ with respect to monomials in the set $S$ are all linearly independent, thereby implying the desired lower bound. 

The third item is an immediate consequence of the observation that the partial derivative complexity of the monomial $\prod_{i = 1}^k x_i$ is equal to $2^k$ and the ``moreover'' part of the first item of this lemma, which says that partial derivative complexity is invariant under invertible linear transformations. 
\end{proof}

We now sketch the main ideas in the proof of \autoref{thm:amir 1}. 
\begin{proof}[Proof of \autoref{thm:amir 1}]
The goal is to prove a lower bound on the partial derivative complexity of the polynomial $\prod_{i = 1}^k q_i(\vx)$. Let $\va \in \F^n$ be the point guaranteed by \autoref{obs:good shift}. Thus,  $q_1(\vx + \va), q_2(\vx + \va), \ldots, q_k(\vx + \va)$ are all zero, and their homogeneous components of degree one are all linearly independent. 
From the first item of \autoref{lem: partial derivative props}, we know that it suffices to prove a lower bound on the partial derivative complexity of $\prod_{i = 1}^k q_i(\vx + \va)$. 

The claim now is that the lowest degree homogeneous component of $\prod_{i = 1}^k q_i(\vx + \va)$ of which is non-zero is the homogeneous component of degree equal to $k$, and this is equal to the product of the homogeneous components of degree one of $q_1(\vx + \va), q_2(\vx + \va), \ldots, q_k(\vx + \va)$. This immediately follows from  \autoref{claim:using homogeneity of g}  in the proof of \autoref{lem:key}.  But once we have this claim, the theorem follows from the third item of \autoref{lem: partial derivative props}. We skip rest of the details. 
\end{proof}
\autoref{thm:amir 2} follows from observing that  we can pick $a_1, a_2, \ldots, a_n$ so that $q_1 + a_1, q_2 + a_2, \ldots, q_k + a_k$ satisfy  Property S. This follows from a similar observation in the proof of \autoref{thm: bj shifted}. Once we have this observation, we are back in the setting of \autoref{thm:amir 1}. 

Before we conclude this section, we note that in order to generalize \autoref{thm:amir 1} and \autoref{thm:amir 2} to completely answer \autoref{q:amir} in affirmative, it would suffice to prove the following conjecture. 
\begin{conjecture}
For all constants $\alpha_1, \alpha_2, \ldots, \alpha_k \in \F$ and linearly independent homogeneous linear forms $\ell_1, \ell_2, \ldots, \ell_k$ the following is true : if $q_1, q_2, \ldots, q_k$ are polynomials such that their minimum degree non-zero monomial has degree at least two, then the partial derivative complexity of the polynomial $\prod_{i = 1}^k (\alpha_i + \ell_i + q_i)$ is at least $2^k$.
\end{conjecture} 
If the conjecture is false, a counterexample to the conjecture may be instructive towards understanding how the partial derivative complexity behaves over taking a product of polynomials. 

\section{Open problems}\label{sec:open q}
We conclude with some open problems. 
\begin{itemize}
\item Perhaps the most natural question would be to characterize the formula complexity of \emph{all} Schur polynomials. As discussed in~\autoref{rem:easy-hard-lambda}, there exist partitions $\lambda$ for which the corresponding $s_\lambda$ have small (polynomial sized) algebraic formulas. On the other hand, as shown in this work, there exist families of $\lambda$s which do not have polynomial sized formulas unless the determinant does. Due to classical results such as~\cite{vsbr83}, we know that the latter class of $\lambda$s in fact have formulas of size $n^{O(\log n)}$. It would be of great interest to extend these two results and get a complete characterization of the formula complexity of $s_{\lambda}$, as a function of the partition $\lambda$. 

\item Bl{\"a}ser and Jindal gave a computationally efficient version of the fundamental theorem for symmetric polynomials. In particular, they showed that if $f_{sym}$ is a symmetric polynomial computed by a polynomial-sized algebraic circuit then the unique polynomial $f_E$ such that $f_{sym} = f_E(e_1, \ldots, e_n)$, also has a polynomial-sized algebraic circuit. 

A natural question one can ask is: if $f_{sym}$ has a polynomial-sized algebraic formula (or ABP) then does $f_E$ also have a polynomial-sized algebraic formula (ABP resp.)? In~\autoref{thm: bj shifted} we take a step towards proving this statement. We show that there exists $\va \in \F^n$ such that if $f_{sym}$ can be expressed as $f_E(e_1-a_1, \ldots, e_n-a_n)$ for a homogeneous $f_E$ then $f_E$ has a small algebraic formula if $f_{sym}$ does. To get the exact Bl{\"a}ser-Jindal-like statement in the formula setting, we would have to improve our result in two ways. We would have to prove it for general $f_E$ rather than for homogeneous $f_E$ and we would have to prove it for $\va=0^n$. We believe that both of these are interesting directions to pursue. 

\item Another interesting extension of the results here would be to  show that there are families of Generalized Vandermonde matrices such that circuit complexity of computing their permanent is essentially as large as the circuit complexity of the Permanent. This would be a VNP analogue of~\autoref{thm:main-intro}. 
\item Yet another interesting direction would be to extend \autoref{thm:amir 1} to answer \autoref{q:amir} completely. 
\end{itemize}

\section*{Acknowledgements}
Mrinal thanks Ramprasad Saptharishi, Amir Shpilka and Noam Solomon and Srikanth thanks Murali K. Srinivasan and Krishnan Sivasubramanian for many insightful discussions. The authors thank Amir for allowing us to include~\autoref{q:amir} and related discussion in this draft.

\bibliographystyle{alphaurlpp}
\bibliography{references}

\end{document}